\RequirePackage{fix-cm}
\documentclass[smallextended]{svjour3}       
\smartqed  
\usepackage{graphicx}
 \usepackage{enumitem}
\usepackage[numbers]{natbib}
\usepackage{amsmath,amssymb,latexsym,graphicx}
\usepackage[normalem]{ulem}
\usepackage{setspace} 
\usepackage{hyperref}
\usepackage{cancel}
\usepackage{scalerel}
\usepackage[dvipsnames,usenames]{color}

\usepackage{relsize}
\newcommand{\ol}[1]{\overline{#1}}

\newtheorem{obs}[theorem]{Observation}
\begin{document}

\title{Identifying species  network features from gene tree quartets under the coalescent model
	\thanks{This research was supported in part by
		the National Institutes of Health grant R01 GM117590,
		awarded under the Joint DMS/NIGMS Initiative to
		Support Research at the Interface of the Biological and
		Mathematical Sciences. }
}

  \titlerunning{Identifying species network features  from gene tree quartets}        

\author{Hector Ba\~nos    
}


\institute{University of Alaska Fairbanks \at
             P.O. Box 756660\\
             Fairbanks, AK 99775-6660 \\               
              \email{hdbanoscervantes@alaska.edu}          
}

\date{}

\maketitle

\begin{abstract}
We show that many topological features of level-1 species networks are identifiable from the distribution of the gene tree quartets under the network multi-species coalescent model. In particular, every cycle of  size at least $4$ and every hybrid node in a cycle of size at least 5 is identifiable. This is a step toward justifying the  inference of such networks which was recently implemented by Sol\'is-Lemus and An\'e. We show additionally how to compute quartet concordance factors for a network  in terms of simpler networks, and explore some circumstances in which cycles of size $3$ and hybrid nodes in 4-cycles can be detected. 	
	
\end{abstract}

\section{Introduction}
\label{intro}
 As phylogenetic analysis of DNA data has progressed, more evidence has appeared showing that  hybridization is often an important factor in evolution. As surveyed in \cite{Nakhleh2011}, hybridization has played a very important role in the evolutionary history of plants, some groups of fish and frogs (\cite{Ellstrand1996}, \cite{Linder2004}, \cite{Mallet2005}, \cite{Noor2006}, \cite{Rieseberg2000}). Other biological processes such as introgression, lateral gene transfer and  gene flow, also require moving beyond a simple tree-like view of species relationships.  
 
Phylogenetic networks are the objects used to represent the relationships between species that admit such events (\cite{Arnold1997},\cite{Bapteste2013}). These networks are often thought of as  obtained from phylogenetic trees by adding additional edges, so that  some nodes in the tree have two parents. Nodes with two parents, called \textit{hybrid nodes}, represent species whose genome arises from two different ancestral species.  Inference of phylogenetics networks from biological data presents new challenges, with methods  still being developed, as shown by recent works including \cite{Ane2007}, \cite{Meng2009}, \cite{Solis-Lemus2016},   and \cite{Yu2011}.

Another challenge in inferring  evolutionary history arises from the fact that many multi-locus data sets exhibit gene tree incongruence, even without suspected hybridization. One possible reason for this is incomplete lineage sorting (ILS), which is described in the tree setting by the multi-species coalescent model \cite{Pamilo1988}.  See for example  \cite{Carstens2007}, \cite{Pollard2006}, and\cite{Syring2005}  where ILS is explained in the biological setting. \\

Meng and Kubatko \cite{Meng2009}   formulated a model of gene tree production, based on the multi-species coalescent model, incorporating both hybridization and ILS. We refer to this model as the \textit{network multi-species coalescent model}, which is further developed in \cite{Nakhleh2012} and \cite{Nakhleh2016}. The model determines the probability of observing any rooted gene tree given a metric rooted phylogenetic species network.
 
An\'e and Sol\'is-Lemus   \cite{Solis-Lemus2016} recently presented a novel statistical method, based on the network multi-species coalescent model, to infer phylogenetic networks from gene tree quartets in a pseudolikelihood framework. The quartets themselves might come from larger gene trees inferred by standard phylogenetic methods.
The pseudolikelihood  in this work is built on  quartet frequencies, or concordance factors, extending  an idea of Liu  \cite{Liu2010} from the tree setting.  The pseudolikelihood  approach is simpler and faster than computing the full likelihood and makes large-scale data analysis more tractable. They demonstrate positive results in reconstructing the  evolutionary relationships among swordtails and platyfishes.

However,  the theoretical underpinnings of the method of \cite{Solis-Lemus2016} are not complete. In using a model for statistical  inference it is important to know if it is theoretically possible to uniquely recover the parameters from the data the model predicts. In more precise terms, for model-based statistical inference to have a solid basis, we need that the probability distribution for data which arises under the model uniquely determines the parameters. This is known as \textit{identifiability} of the model parameters.

While \cite{Solis-Lemus2016} highlighted important issues of parameter identifiability needed to justify its inference method, it included only preliminary investigations. The authors aptly argue that   gene quartet probabilities can be computed for larger networks and carry some information about network topology and edge lengths, but do not formally provide a full proof. They argue that some hybridization can be detected, but  did not establish  what features in a large network topology can be identified. Working in the setting of level-1 networks, which is also adopted here, their investigations into parameter identifiability focus on small networks of 4 or 5 taxa.

The primary purpose of this work is to begin to address some of these  identifiability questions raised in \cite{Solis-Lemus2016}. That is, we study the question: given information on gene quartet probabilities for some unknown level-1 network $\mathcal{N}$, what can be determined about the topology of $\mathcal{N}$? We limit our focus to topological features of networks, leaving hybrid parameters, and metric identifiability for subsequent study. 

Although others have considered the problem of constructing large networks from small ones, these works do not seem to be applicable to the question studied here. Most of these works, including \cite{Huber2015} and  \cite{Keijsper2014}, are primarily combinatorial in nature. In particular, these studies do not address ILS through the network multi-species coalescent model, nor the types of inputs that might be obtained from biological data.\\ 

The main result of this work, Theorem \ref{thm::main} of Section \ref{sec final}, is that under the network multi-species coalescent model on level-1  networks, we can generically identify from gene quartet distributions ``most" of the unrooted topological network, including all cycles of size at least $4$, and hybrid nodes in the cycles of size greater than 4, from quartet gene tree distributions. ``Generically" here means for all values of numerical parameters except those in a set of measure zero. The methods used are a mix of the semi-algebraic study of quartet gene tree frequencies  (in terms of linear equalities  and inequalities they satisfy) with combinatorial approaches to combining this knowledge for many quartets. As a side benefit the proofs suggest combinatorial methods for inferring   networks. However, we do not explore how such methods might be implemented in the presence of the noise that any collection of inferred gene trees will have. 

Another result of this work, in Section \ref{Q CF sn}, is a rigorous derivation of  how gene quartet probabilities can be computed  for large networks under the  coalescent model. Although this parallels some of the results in \cite{Solis-Lemus2016}, the   proofs given here deal with   complications concerning passing from large rooted networks  to unrooted quartet networks that were left unaddressed in that work. This is accomplished by expressing quartet frequencies as convex combinations of those on simplified networks, ultimately leading to expressions in terms of trees.\\

The outline of this work is as follows: Section 2  introduces basic definitions and establishes some terminology on graphs and networks. Section 3 sets forth insights and tools for studying the structure of level-1 networks. Section 4 reviews the network multi-species coalescent model  of \cite{Meng2009}, as well as quartet  concordance factors and some of their properties.   In Section 5 we show how concordance factors of quartet networks can be expressed in terms of simpler networks. Section 6 introduces the ``Cycle  property" of concordance factors and Section 7 defines the ``Big Cycle"  property of concordance factors. In Section 8, the main result on topological network identifiability is proved using the Big Cycle property and in  Section 9 some extended results  on the ``Cycle property" are shown.


\section{Phylogenetic networks}\label{sec def}
We adopt standard terminology for graphs and networks, as used in phylogenetics; see for example \cite{Semple2005} and \cite{Steel2016}.  All undirected, directed, or semidirected graphs will not contain loops. If $G$ is a directed or semidirected graph, the \textit{undirected graph of $G$}, denoted by $U(G)$, is the graph $G$ with all directions  omitted. 
\subsection{Rooted networks}
To set terminology, we begin with some fundamental definitions.
\begin{definition}\label{def::network} 
	A \textbf{topological binary rooted phylogenetic network} $\mathcal{N}^+$
	on taxon set $X$ is a connected directed acyclic graph with vertices $V = \{r\} \sqcup V_L \sqcup V_H \sqcup V_T$, edges $E = E_H \sqcup E_T$ and a bijective leaf-labeling function $f : V_L \to X$ with the
	following characteristics:
	\begin{itemize}
		\item[1.]The \textbf{root} $r$ has indegree 0 and outdegree 2.
		\item[2.] A \textbf{leaf} $v \in V_L$ has indegree 1 and outdegree 0.
		\item[3.] A  \textbf{tree node} $v\in  V_T$ has indegree 1 and outdegree 2.
		\item[4.] A \textbf{hybrid node} $v\in  V_H$ has indegree 2 and outdegree 1.
		\item[5.] A \textbf{hybrid edge} $e \in E_H$ is an edge whose child is a hybrid node.
		\item[6.] A \textbf{tree edge} $e \in E_T$ is an edge whose
		child is a tree node or a leaf.
	\end{itemize}
\end{definition}

\begin{definition}
	Let $\mathcal{N}^+$ be a topological binary rooted phylogenetic network with  $|E|=m$ and  $|E_H|=2h$. A \textbf{metric for $\mathcal{N}^+$} is a pair $(\lambda, \gamma)$, where
	$\lambda:E\to \mathbb{R}_{> 0}$ and  $\gamma:E_H\to  (0,1)$ satisfies that if two edges $h_1$ and $h_2$ have the same hybrid node as child, then $\gamma(h_1)+\gamma(h_2)=1$. 
	
	If  $(\lambda, \gamma)$ is a metric for $\mathcal{N}^+$, then we refer to $(\mathcal{N}^+,(\lambda, \gamma))$ as a \textbf{metric binary rooted phylogenetic network}. 
\end{definition}
Note that Definition \ref{def::network} differs from that of \cite{Steel2016} in that it allows up to two edges between a pair of nodes. An edge weight $\lambda(e)$ is interpreted as the time (in coalescent units) between speciation events represented by the ends of edge $e$.  For any hybrid edge $h$ with child $v$, the value $\gamma(h)=\gamma_{h}$  is the probability that a lineage at $v$ has ancestral lineage in $h$ and is called  \textit{hybridization parameter}. \\

\subsection{Most recent common ancestor}
We generalize the concept of the most recent common ancestor of a set of taxa on trees to the network setting. 

\begin{definition}\label{def::above}
	Let $\mathcal{N}^+$ be a (metric or topological) binary rooted phylogenetic network. We say that a node $v$ is \textbf{above} a node $u$, and $u$ is \textbf{below} $v$, if there exists a non-empty directed path in $\mathcal{N}^+$ from $v$ to $u$. We also say that an edge with parent node  $x$ and   child   $y$ is above (below) a node $v$ if  $y$ is above or equal to $v$   ($x$ is below or equal to $v$).
\end{definition}
 Note that since $\mathcal{N}^+$ has no directed cycles, $u$ cannot be both above and below $v$. 
\begin{definition}\label{def::MRCA}
	Let $\mathcal{N}^+$ be a (metric or topological)  binary rooted phylogenetic network on $X$  and let $Z\subseteq X$. Let $D$  be the set of  nodes which lie on every directed path from the root $r$ of $\mathcal{N}^+$ to any $z\in Z$. Then the \textbf{most recent common ancestor of $Z$ of $\mathcal{N}^+$}, denoted by \emph{MRCA}$(Z,{\mathcal{N}^+})$,  is the unique node $v\in D$  such that $v$ is below all $u\in D$, $u\neq v$.
\end{definition}

When $\mathcal{N}^+$   is clear from  context, we  write MRCA$(Z)$ for MRCA($Z,{\mathcal{N}^+}$). 
To see that MRCA($Z$) is well defined for any $Z\subseteq X$, note first that $D\neq\emptyset$ since  $r\in D$. Also, since every pair of nodes $u,v\in D$ both lie on a path, we have a notion of above and below for $u$ and $v$, i.e.  a total order on $D$, and hence a minimal element. \\
 
Note that this definition of network MRCA differs from that of the least common ancestor (LCA) that appears elsewhere in the phylogenetic network literature \cite{Steel2016}. While it agrees with the usual concept for trees, it is more subtle on networks. In particular, if $\mathcal{N}^+$ is a network on $X$, MRCA($X$) need not to be the root of the network, as Figure \ref{fig::root2cyc} (left) shows. Furthermore, there can be nodes below the MRCA($X$) which are ancestral to all of $X$, as Figure \ref{fig::ancestralofall} shows.
 \begin{figure}\begin{center}
		\includegraphics[scale=.49]{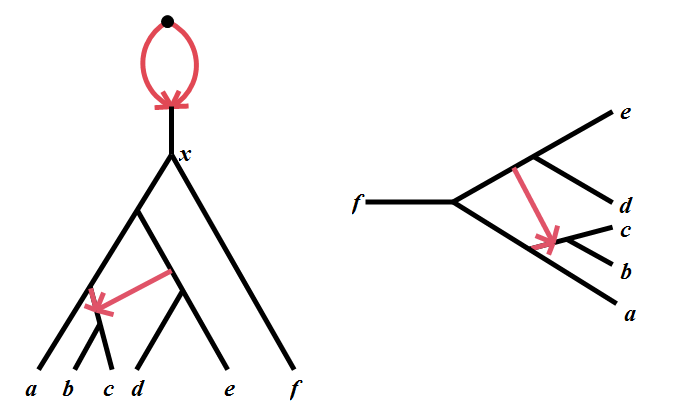}  
		\caption{(Left) A binary rooted phylogenetic network on $X$, with MRCA$(X)$ the node labeled $x$, and (Right) its induced unrooted semidirected network. In a depiction of a rooted network, all edges are directed downward, from the root, but arrowheads are shown only on hybrid edges. For the unrooted network, all edges except hybrid ones are undirected.}\label{fig::root2cyc}
\end{center}\end{figure}
 \begin{figure}\begin{center}
 		\includegraphics[scale=.35]{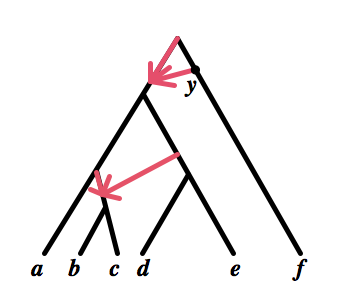}  
 		\caption{A binary rooted phylogenetic network where the node labeled $y$ is ancestral to all taxa in $X$ but is not  MRCA($X$).  MRCA($X$) here is the root of the network.  }\label{fig::ancestralofall}
 \end{center}\end{figure}
 \begin{lemma}\label{lem::MRCAproperties}
	Let $\mathcal{N}^+$ be a (metric or topological) binary rooted phylogenetic network on $X$  with root $r$, and let  $Z\subseteq Y\subseteq X$. Then
	\begin{enumerate}
	    	\item[(i)] the indegree of MRCA$(Z)$ is at most one for any $Z\subset X$;
			\item[(ii)] at most one of the out edges of MRCA($Z$) is hybrid;
			\item[(iii)] if $Z\subseteq Y\subseteq X$ then  MRCA($Z$) is below or equal to MRCA($Y$).
	\end{enumerate}
\end{lemma}
\begin{proof} To see (i), suppose that the indegree  of MRCA$(Z)$ is two. Then the outdegree would be one, and the child of MRCA$(Z)$ would be in any path from the root to any taxa in $Z$,  contradicting the definition of MRCA$(Z)$. 
	
	For (ii), suppose the  out edges of MRCA($Z$), $e_1$ and $e_2$, are both hybrid. If  $e_1$ and $e_2$ have the same child then every path from   $r$   to any $z\in Z$ would contain that node, contradicting the definition of MRCA($Z$).

	 Now denote by $x_1\neq x_2$ the child nodes of $e_1$ and $e_2$ respectively.  If both $x_1$ and $x_2$ had parents below MRCA($Z$), then $x_1$ has a parent below $x_2$ and $x_2$ has a parent below $x_1$ giving a directed cycle. Thus, without loss of generality, assume $x_1$ has parents MRCA($Z$) and $v$ with $v$ not below MRCA($Z$).   Let $z\in Z$ with $z$ below $x_1$. If we remove the   MRCA($Z$) from $\mathcal{N}^+$ there is still a path from $r$ to $z$ (which goes from $r$ to $v$ to $x_1$ to $z$). This contradicts the fact that MRCA($Z$) is on all paths from $r$ to any $z\in Z$.
	 
	 Finally, (iii) follows directly from the definition.  
 $\Box$\end{proof}

\begin{lemma}\label{lem::mrcaxy}
	Let $\mathcal{N}^+$ be a (metric or topological)   binary rooted phylogenetic  network on  $X$  and let $Z\subset X$, $|Z|\geq 2$. For every  $x\in Z$, there is a $y\in Z$ such that MRCA($x,y$)=MRCA($Z$).
\end{lemma}
\begin{proof}
	Let $m$=MRCA($Z$), fix $x\in Z$ and let $P$ be a path from $m$ to $x$. By definition of MRCA, for all $y\in Z$, MRCA($x,y$) is a node in $P$ and is below or equal to $m$ by Lemma \ref{lem::MRCAproperties}. Suppose that  MRCA($x,y$) is below $m$ for all $y\in Z$. Let $z\in Z$ be such that  MRCA$(x,z)$ is above or equal to MRCA($x,y$) for all $y\in Z\smallsetminus\{z\}$. 
	
		We claim that any path from $m$ to   $y\in Z$ passes through  MRCA$(x,z)$. Suppose there exists taxon $y$  with path $P'$ from $m$ to $y$ that does not pass through  MRCA$(x,z)$.  But $P'$ must pass through MRCA($x,y$). Since MRCA($x,y$) is below  MRCA$(x,z)$, there is a path from $m$ to MRCA($x,y$) to $x$ that does not contain MRCA($x,z$). This is a contradiction. 
		
			But every path from $m$ to any $y\in Z$ passes through MRCA($x,z$), contradicting that MRCA($x,z$) is below $m$.  $\Box$\end{proof}

By this Lemma  we can characterize the MRCA($Z $) as the highest node of the form MRCA($x,y$) for some $x,y\in Z$ , or  the highest node of that form for fixed $x\in Z$. 


\subsection{Unrooted networks.}
Let $G$ be a  directed or semidirected  graph with $z$   a degree two node.  Let $x$ and $y$ be the two nodes adjacent to $z$. Then, up to isomorphism, the subgraph on  $x,y$ and $z$ must be one of the graphs shown on the left of Figure \ref{fig::supress}, which we denote by $H$. By \textit{suppressing} $z$ we mean replacing $H$ in $G$ by the graph to the right of it in Figure \ref{fig::supress}.  
\begin{figure}\begin{center}
	\includegraphics[scale=.45]{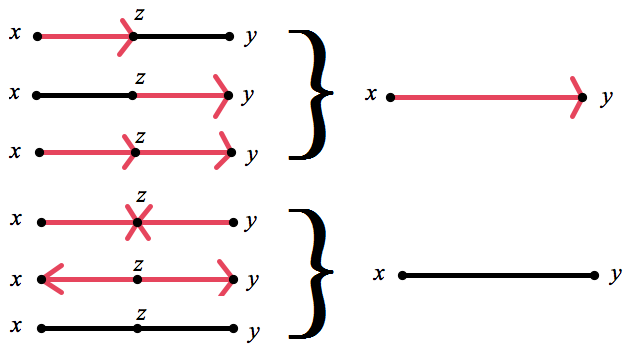}  
	\caption{On the left are all the semidirected graphs, up to isomorphism, on a degree two node $z$ and its adjacent vertices $x$ and $y$. On the right are the corresponding graphs obtained by suppressing $z$.}\label{fig::supress}
\end{center}\end{figure}	
\begin{definition}	
	Let $\mathcal{N}^+$ be a  binary topological rooted phylogenetic network on a set of taxa $X$. Then  $\mathcal{N}^-$ is the semidirected network obtained by 1) keeping only the edges and nodes below MRCA($X$); 2) removing the direction of all tree edges; 3) suppressing  MRCA($X$). We refer to $\mathcal{N}^-$ as the \textbf{topological unrooted semidirected network induced from $\mathcal{N}^+$}.
 \end{definition}	
Figure \ref{fig::root2cyc} shows an example of a network $\mathcal{N}^+$ and its induced $\mathcal{N}^-$. We now introduce a metric on $\mathcal{N}^-$ induced from one on $\mathcal{N}^+$.

\begin{definition}\label{def::urooted}
	Let $(\mathcal{N}^+,(\lambda, \gamma))$ be a metric binary rooted phylogenetic network and let  $\mathcal{N}^-$ be the topological unrooted semidirected network induced from $\mathcal{N}^+$.
	Denote by $e^*$ the edge of $\mathcal{N}^-$ introduced in place of  the edges $e_1$ and $e_2$ in $\mathcal{N}^+$ when MRCA$(X)$ is suppressed.  Define $\lambda':E(\mathcal{N}^-) \to \mathbb{R}_{>0}$ such that $\lambda(e^*)=\lambda(e_1)+\lambda(e_2)$ and $\lambda'(e)=\lambda(e)$ for  $e\in \mathcal{N}^-$, $e\neq e^*$. If $e^*$ is not   hybrid, $\gamma'=\gamma$, else let $\gamma'(h)=\gamma(h)$ for all hybrid edges of $\mathcal{N}^-$ other than $e^*$ and  $\gamma'(e^*)=\gamma(e_i)$, where $e_i$ is, by Lemma \ref{lem::MRCAproperties}, the single hybrid edge in $\{e_1,e_2\}$.  We refer to $(\mathcal{N}^-,(\lambda', \gamma'))$ as the \textbf{metric unrooted semidirected network induced from $(\mathcal{N}^+,(\lambda, \gamma))$}. 	
\end{definition}

The networks considered in this work are always induced  from a rooted binary metric phylogenetic network.  To simplify language, we refer to a (metric or topological) binary rooted phylogenetic network as a \emph{(metric or topological) rooted network} and  to a induced (metric or topological) unrooted semidirected  phylogenetic network as a  \emph{(metric or topological) unrooted network}. 

We note that not all binary semidirected graphs are topological unrooted networks, since some graphs are not compatible with suppressing the root on any rooted network. Moreover, $\mathcal{N}^-$ might be induced from several rooted networks $\mathcal{N}^+$. See Figure \ref{Fig:Induced and not}.

\begin{figure}\begin{center}
	\includegraphics[scale=.35]{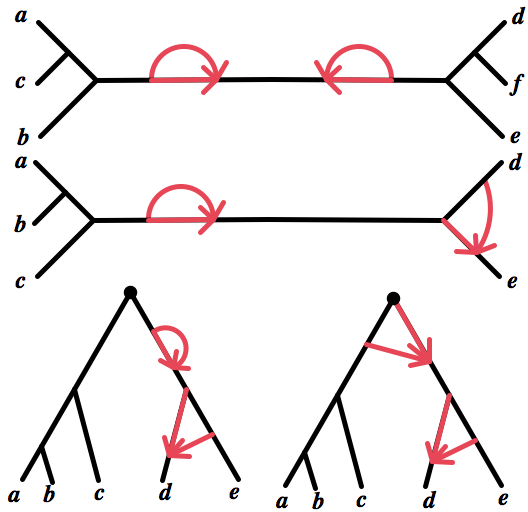}  
	\caption{The top graph is not a topological unrooted semidirected
		phylogenetic network, since its  directed edges  cannot be obtained by suppressing the root of any 6-taxon topological binary rooted phylogenetic network. The middle graph is the induced topological unrooted network from either of the bottom rooted networks, as well as others.}\label{Fig:Induced and not}
\end{center}\end{figure}

Although an unrooted  network $\mathcal{N}^-$ does not have a root specified, since  hybrid edges are directed, the suppressed MRCA($X$) of $\mathcal{N}^+$ must have been located `above' them. Thus  in $\mathcal{N}^-$, we still have a well-defined notion of which taxa are descendants of a hybrid node $v$. These are  the taxa $x$ such that  there exists a semidirected path from $v$ to $x$ in $\mathcal{N}^-$. In this case we say that \textit{$x$ descends from $v$}. 

\subsection{Induced networks on subset of taxa}
Since later arguments require   an understanding  of the behavior of the network multi-species coalescent model on a subset of taxa, we introduce some needed definitions.

\begin{definition}\label{def:inducerootnetontaxa}
	Let $\mathcal{N}^+$ be a (metric or topological) rooted network on  $X$  and let $Z\subset X$. The \textbf{induced rooted network $\mathcal{N}^+_Z$ on $Z$}  is the network obtained from  $\mathcal{N}^+$ by 1) retaining only
	edges and nodes in  paths from the root to any  taxa in $Z$; 2) suppressing all degree two nodes  except the  root; 3) in the case the root then has outdegree one,  contracting the edge incident to the root. 
\end{definition}
Note that MRCA($Z,\mathcal{N}^+_Z$)=MRCA($Z,\mathcal{N}^+$).  If $|Z|=4$ then  $\mathcal{N}^+_Z$, the  \textit{induced rooted quartet network on $Z$}, will also be  denoted  by  $\mathcal{Q}^+_Z$ to emphasize it involves only 4 taxa.

\begin{definition}\label{def:inducemrcanetontaxa}
	Let $\mathcal{N}^+$ be a (metric or topological) rooted network on  $X$  and let $Z\subset X$. The \textbf{induced MRCA network of $Z$}, denoted $\mathcal{N}^\oplus_Z$, is the rooted network obtained from $\mathcal{N}^+_Z$ by deleting everything above the MRCA($Z,\mathcal{N}^+$)
\end{definition}

In particular we note that $\mathcal{N}^\oplus_Z$ has root  MRCA($Z,\mathcal{N}^+$). If $|Z|=4$ then  $\mathcal{N}^\oplus_Z$, the  \textit{induced MRCA quartet network on $Z$}, is also denoted  by  $\mathcal{Q}^\oplus_Z$.\\ 

\begin{definition}
	Let $G$ be a semidirected graph and let $x,y$ be two nodes in $G$. A \textbf{trek} in $G$ from $x$ to $y$ is an ordered pair of semidirected
	paths $(P_1, P_2)$ where $P_1$ has terminal node $x$, $P_2$ has terminal node $y$, and both $P_1$ and $P_2$
	have starting node $v$. The node $v$  is called the \textbf{top} of the trek, denoted   top$(P_1,P_2)$. A trek $(P_1, P_2)$ is \textbf{simple} if the
	only common node among $P_1$ and $P_2$ is  $v$.  
\end{definition}

This definition is adopted from non-phylogenetic studies of statistical models on graphs, such as \cite{sullivant2010}.

\begin{definition}\label{def:inducenetontaxa}
	Let $\mathcal{N}^-$ be a (metric or topological) unrooted network on  $X$  and let $Z\subseteq X$. The \textbf{induced unrooted network $(\mathcal{N}^-)_Z$ on a set of taxa $Z$}  is the network obtained from  $\mathcal{N}^-$ by retaining only  edges in simple treks between pairs of taxa in $Z$, and then suppressing all degree two nodes.
\end{definition}

If $|Z|=4$ then  $(\mathcal{N}^-)_Z$, the  \textit{induced unrooted quartet network on $Z$}, is also denoted  by  $\mathcal{Q}^-_Z$. 


While the following statement is intuitively plausible, its proof is rather involved and thus given in the Appendix.

\begin{proposition}\label{prop::ancestralinduced}
	Let $\mathcal{N}^+$ be a (metric or topological) rooted network on $X$ and  let    $Z\subseteq X$. Then  $(\mathcal{N}^-)_Z$ and   $(\mathcal{N}^+_Z)^-$ are isomorphic.
\end{proposition}

\subsection{Cycles}
Although the networks $\mathcal{N}^+$, $\mathcal{N}^-$ are acyclic (in both, the directed and semidirected settings), their undirected graphs $U(\mathcal{N}^+)$, $U(\mathcal{N}^-)$ may contain a cycle. Thus the term `cycle' may be used to unambiguously refer to cycles in the undirected graphs. We formalize this with the following definition:

\begin{definition}\label{cycle}
	Let $\mathcal{N}$ be a (metric or topological, rooted or unrooted) network. A \textbf{cycle} in $\mathcal{N}$ is a non-empty  path from a node to itself,  allowing edges to be traversed without regard to their possible direction. 
	The \textbf{size} of the cycle is the number of edges in the path. A \textbf{$k$-cycle} is a cycle of size $k$. 
\end{definition}

By \textit{suppressing a cycle} $C$ in a graph we  mean removing all edges in $C$ and identifying all nodes in $C$.

\section{Structure of level-1 networks}

The class of  phylogenetic networks is often too large to obtain  strong mathematical results, so it is common to restrict to networks that have a simpler structure, for instance, the class of \textit{level-1}  phylogenetic networks.

\begin{definition}\label{level1}Let $\mathcal{N}$ be a (rooted or unrooted) topological   network. If no two cycles in $\mathcal{N}$ share an edge, then $\mathcal{N}$ is  \textbf{level-1}.
\end{definition}
If  $\mathcal{N}$ is a level-1 network, any subnetwork or induced network of $\mathcal{N}$ is also level-1. \\

Given a hybrid node $v$,  denote the hybrid edges whose child is $v$ by $h_v$ and $h_v'$. Then  $h_v$ and $h_v'$ are called the \textit{hybrid edges of} $v$.

\begin{lemma}\label{lem:assos}
	Let $\mathcal{N}$  be a (topological or metric, rooted or unrooted) level-1 network and let $C$ be a cycle of $\mathcal{N}$. Then $C$ contains exactly one hybrid node $v$, and the associated hybrid edges $h_v$, $h_v'$. Furthermore, each node of $\mathcal{N}$ is in at most one cycle and, as a result, $v$, $h_v$ and $h_v'$ are in exactly  one cycle of $\mathcal{N}$. 
\end{lemma}

The proof of each statement of this Lemma, using different terminology, is given by Rossello and Valiente \cite{Rossello2009}. 

\begin{proposition}\label{prop::structureabove}
	Let $\mathcal{N}^+$ be a   topological level-1 rooted network on   $X$.  
	The structure of all the nodes and edges above MRCA($X$) in $\mathcal{N}^+$ is a (possibly empty) chain of 2-cycles connected by edges, as depicted in   Figure \ref{fig::aboveMRCAZ}.
\end{proposition}
\begin{proof}
	Let $m=\text{MRCA}(X)$, and denote by $r$ the root of $\mathcal{N}^+$. The proof is by induction on the number of the edges above $m$.
	If there are no edges above $m$, then $m=r$ and the result is trivially true. By Lemma \ref{lem::MRCAproperties}, one easily sees that there cannot be only $1$ or $2$ edges above $m$ in a binary phylogenetic network.
	
	Now assume the claim holds when there are at most $k$ edges above $m$ and suppose there are $k+1$ edges above $m$. Note that $r$ has outdegree 2 by the definition of  $\mathcal{N}^+$.  
	
	Suppose that edges incident to $r$ have different children,  $x$ and $y$. Note neither $x$ nor $y$ can be $m$ . The outdegree of one of $x$ or $y$ must be 2, otherwise both would be hybrid nodes, which would require $x$ above $y$ and $y$ above $x$. Without loss of generality suppose  $x$ has outdegree 2, and denote by $e_1$ and $e_2$ its out edges, and denote by $e_3$ the edge $(r,y)$.  Since every path from $r$ to a leaf goes through $m$, there are at least 3 distinct paths $P_1$, $P_2$, $P_3$ from $r$ to $m$, where $P_i$ contains $e_i$.

	  This contradicts the level-1 condition. Thus $x=y$, and the edges from $r$  form a 2-cycle.

	Now since $x$ is a hybrid node, it has outdegree 1, with child $v$.  Also,  there are $k-3$ edges above $m$ that are also below $v$.  Applying the inductive hypothesis to   $\mathcal{N}^+$ with edges above $v$ removed, the result follows.  $\Box$\end{proof}

\begin{figure}\begin{center}
	\includegraphics[scale=.35]{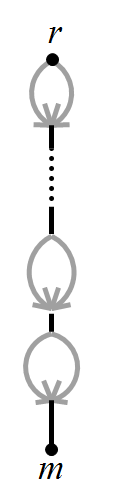}  
	\caption{In a level-1 network on $X$, the structure between the root and $m=$MRCA($X$) is a chain of two cycles. The number of two cycles in the chain could be zero.}\label{fig::aboveMRCAZ}
\end{center}\end{figure}

 Proposition \ref{prop::structureabove} applied to $\mathcal{N}^+_Z$ illustrates the structure of the common ancestry of a subset $Z$ of taxa.   When we pass to a MRCA network or an induced unrooted network,  we ``throw away"  this structure. We   show in Section \ref{Q CF sn}  that under the network multi-species coalescent model this  structure  has no effect on the formation of quartet gene trees. \\

Let $v$ be a hybrid node in a level-1 (rooted or unrooted, metric or topological) network $\mathcal{N}$ on  $X$   and let  $C_v$ be the cycle containing $v$. By removing $C_v$ from $\mathcal{N}$ we obtain a partition of $X$ according to the connected components of the resulting graph. We refer to this partition as the \textit{$v$-partition} and  its partition sets as $v$\textit{-blocks}. 

Note that each node in $C_v$ can be associated to a $v$-block. That is, a $v$-block $B_u$ is associated to a node $u$ in $C_v$ if by removing $u$ from the network, the induced partition of taxa is $\{B_u,X\smallsetminus B_u\}$. We refer to the $v$-block $B_v$, whose elements descend from $v$, as the \textit{${v}$-hybrid block}. Two distinct $v$-blocks $B_u,B_w$ are \textit{adjacent} if the nodes $u,w\in C_v$ are adjacent.

Let $\mathcal{D}=\{C_1,...,C_n\}$ be a collection of  cycles in $\mathcal{N}$.
 The partition of $X$ obtained by removing all cycles in $\mathcal{D}$ is the \textit{network partition induced by  $\mathcal{D}$} and its blocks are \textit{network blocks induced by  $\mathcal{D}$}. When  $\mathcal{D}$  is the set of all cycles  in $\mathcal{N}$ of size at least $k$, the partition is the \textit{$k$-network partition} and  its blocks are   \textit{$k$-network blocks}. The $4$-network blocks   play an important role in Section \ref{sec final}.
 
  The following is straightforward to prove.


\begin{lemma}\label{2taxnetblo}Let $\mathcal{N}$ be a level-1  (rooted or unrooted)  topological network on  $X$. Let $\mathcal{D}=\{C_1,...,C_n\}$ be a collection of  cycles in $\mathcal{N}$.
	For any two taxa $a$ and $b$  in different network blocks  induced by  $\mathcal{D}$, there exists a hybrid node $v$ of some cycle in $\mathcal{D}$ such that $a$ and $b$ are in different $v$-blocks.
\end{lemma}

%
%

If two taxa $a$ and $b$ are in the same network block induced by $\mathcal{D}$, then they are connected when all cycles in $\mathcal{D}$ are removed, and hence when any single cycle in $\mathcal{D}$ is removed. This comment together with Lemma \ref{2taxnetblo} yields the following.
 
\begin{corollary}\label{cor:partiofparti}
	Let $\mathcal{N}$ be a level-1 (rooted or unrooted)  topological  network on   $X$. Let $\mathcal{D}=\{C_1,...,C_n\}$ be a collection of  cycles in $\mathcal{N}$, with $v_i$  the hybrid node associated to $C_i$. The network partition induced by  $\mathcal{D}$ is the common refinement of the $v_i$-partitions for  $1\leq i\leq n$.
\end{corollary}


Since suppressing cycles in level-1 networks does not introduce loops or multi-edges, we can define a notion of a tree  of cycles which is useful for the proof of Theorem \ref{thm::main}.

\begin{definition}\label{treeofcycles}
	Let $\mathcal{N}^-$ be a    topological unrooted level-1 network. 
	Let $\mathcal{T}$ be the graph obtained from $\mathcal{N}^-$ by	 1) removing all pendant edges, repeatedly, until  no pendant edges remain; 2) suppressing all vertices of degree two that  are not part of a cycle; 3) suppressing all cycles in the resulting level-1 network comprised  of cycles joined by some edges. We refer to $\mathcal{T}$ as the \textbf{tree of  cycles of $\mathcal{N}^-$}.
\end{definition} 

In the tree of  cycles of $\mathcal{N}^-$ certain nodes, including all the leaves, represent a cycle of the original network $\mathcal{N}^-$. The notion of tree of cycles is similar to but different from ``tree of blobs" of \cite{Gusfield2007}. In  Figure \ref{fig::treeofcyc} we see an example of a tree of cycles.

\begin{figure}\begin{center}
	\includegraphics[scale=.28]{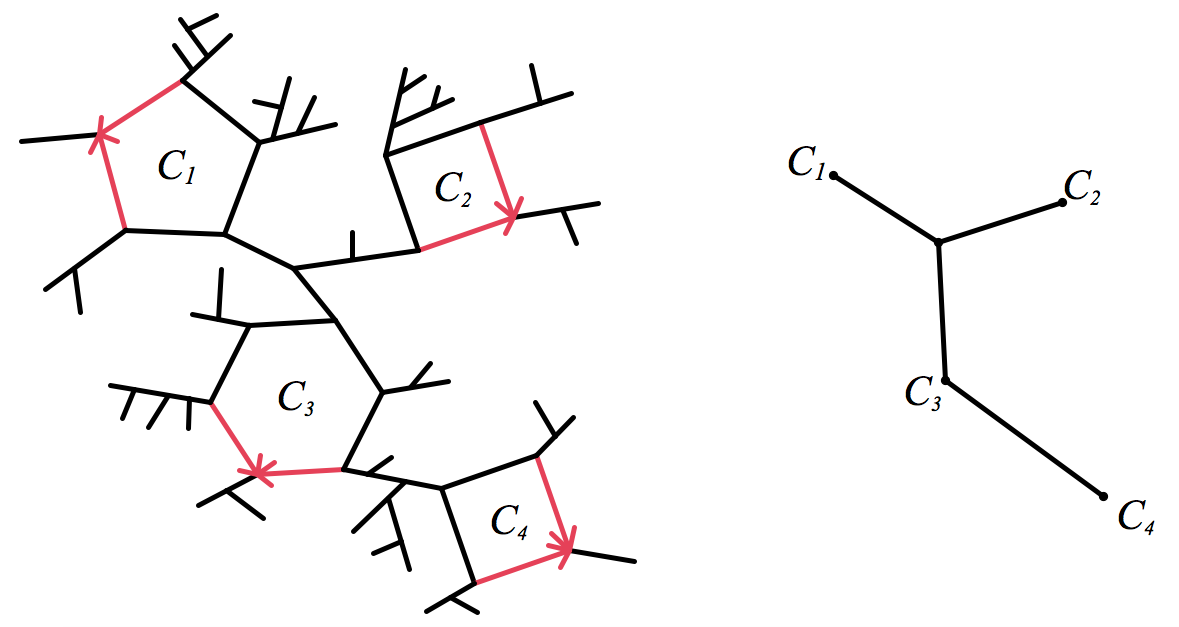}  
	\caption{(Left) A level-1 unrooted network $\mathcal{N}^-$ and (Right) the tree of  cycles of $\mathcal{N}^-$.}\label{fig::treeofcyc}
\end{center}\end{figure}

\section{The network multi-species coalescent model and quartet concordance factors.}
Coalescent theory models the formation of gene trees within populations of species. The coalescent model for a single population traces (backwards in time) the ancestries of a finite set of individual copies of a gene as the lineages \emph{coalesce} to form ancestral lineages (see \cite{Wakeley2008}). The \textit{multi-species coalescent (MSC) model}  is a generalization of the coalescent model, formulated by applying it to multiple populations connected to form a rooted population tree, or species tree. It is   commonly used to obtain the probabilities of gene trees in the presence of incomplete lineage sorting.\\

Meng and Kubatko  \cite{Meng2009} extended the MSC  by introducing phenomena such as hybridization or other horizontal gene transfer across the species-level and Nakhleh et al. further developed it \cite{Nakhleh2012,Nakhleh2016}.  This model describes any situation in which a gene lineage may ``jump" from one population to another at a specific time. The model parameters are specified by a metric binary rooted phylogenetic network as defined in Section \ref{sec def}.
Different from  models such as the structured coalescent with continuous gene flow (see \cite{Wakeley2008}), the network model approach assumes the gene transfer occurs at a single point in time along hybrid edges. We refer to this  extended version of the MSC   as the \textit{network multi-species coalescent (NMSC) model}.

The NMSC model  assumes that speciation by hybridization results in  what Meng and Kubatko refer to as a mosaic genome. 
One assumption of the  NMSC model, inherited from the MSC model, is that all gene lineages present at a specific point on the species tree behave identically above this point. More precisely, the conditional probability of any outcome of the coalescent process above this point is invariant under permutations of those lineages. This feature is  known as the \textit{exchangeability} property.


\begin{figure}\begin{center}
	\includegraphics[scale=.25]{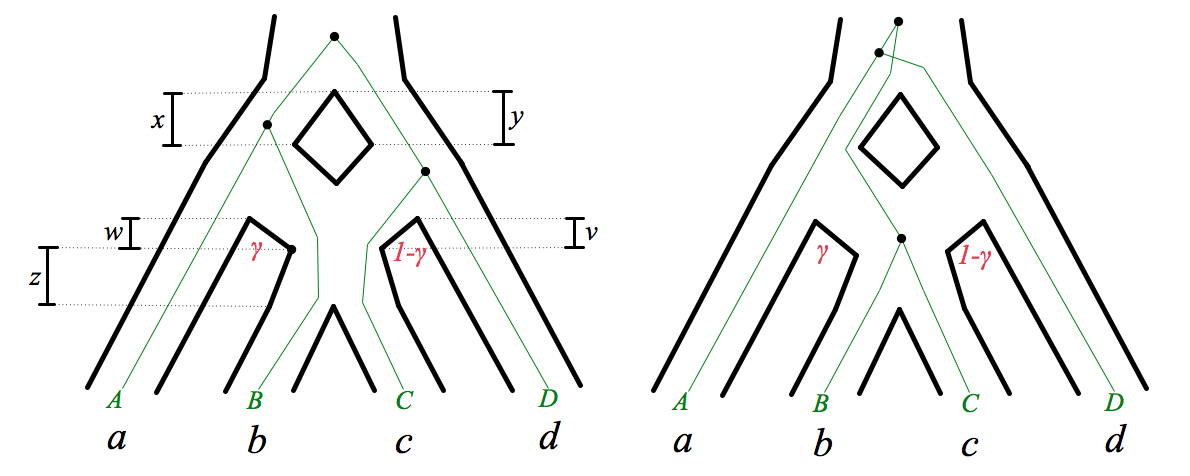}
	\caption{Two gene trees within a species network with one hybrid node.}\label{fig:NMSCgene}
\end{center}\end{figure}

\begin{example}\label{ex::nmsc}
	We illustrate how to compute the probability of a gene tree topology under the NMSC with an example.  Suppose we have the rooted metric species network given on  Figure \ref{fig:NMSCgene}. Let $A,B,C$ and $D$ be genes sampled from species $a,b,c$ and $d$ respectively. We compute the probability that a gene tree has the unrooted topology $((A,B),(C,D))$ under the NMSC model.  \\
	
	First observe that until $B$ and $C$ trace back to the edge with length $z$ there cannot be a coalescent event. In that edge these lineages cannot coalesce if the gene tree $((A,B),(C,D))$ is to be formed. The probability of no coalescence on this edge is $e^{-z}$. Now there are 4 cases, illustrated in Figure \ref{fig:exanmsc}:
		\begin{enumerate}
		\item[1)] with probability $\gamma^2$, lineages  $B$ and $C$ enter the  edge of length $w$;\\
		\item[2)] with probability $(1-\gamma)^2$,  $B$ and $C$ enter  the  edge  of length $v$;\\
		\item[3)] with probability $\gamma(1-\gamma)$, $B$ enters the edge  of length $w$ and  $C$ enters the  edge  of length $v$;\\
		\item[4)] with probability $(1-\gamma)\gamma$, $B$ enters the edge  of length $v$ and  $C$ enters the  edge  of length $w$. 
	\end{enumerate}
 
	Observe that  each case is now reduced to a standard MSC scenario with several samples per population (see \cite{Degnan2010}). Let $P_i$ the probability of observing $((A,B),(C,D))$ under the MSC of case $i$.
	Then the probability of observing $((A,B),(C,D))$ is $e^{-z}(\gamma^2 P_1+(1-\gamma)^2 P_2+\gamma(1-\gamma)P_3+\gamma(1-\gamma)P_4)$.  
	\begin{figure}\begin{center}
		\includegraphics[scale=.305]{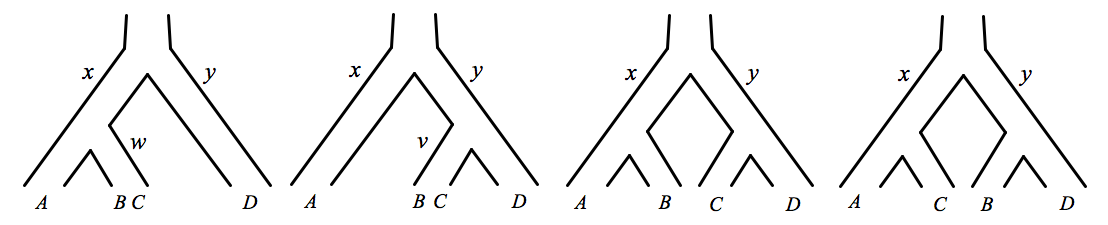}
		\caption{Cases 1-4 (Left-Right) of Example \ref{ex::nmsc}, of how lineages may behave under the NMSC model on the network of Figure \ref{fig:NMSCgene}.}\label{fig:exanmsc}
	\end{center}\end{figure}
\end{example}

Following  An\'e and Sol\'is-Lem\'us  \cite{Solis-Lemus2016},  we are interested in the probability that a species network produces various gene quartets under the NMSC. This motivates the following definition.

\begin{definition}
	Let  $\mathcal{N}^+$ be a metric rooted network on a taxon set $X$. Let $A,B,C,D$ be genes sampled from species $a,b,c,d$ respectively. Given a gene quartet $AB|CD$, the \textit{quartet concordance factor} $CF_{AB|CD}$ is the probability under the NMSC on $\mathcal{N}^+$ that a gene tree displays the quartet $AB|CD$, and  $$CF_{abcd}=(CF_{AB|CD},CF_{AC|BD},CF_{AD|BC})$$ is the ordered triplet of  concordance factors  of each quartet on the taxa $a,b,c,d$.
\end{definition}


In the particular case where $\mathcal{N}^+$ has no hybrid edges, so the network is a tree, it is known that the quartet concordance factors do not depend on the root placement \cite{Allman2011}.  For example let $a,b,c,d$ be taxa and consider any root placement in the unrooted species tree with topology $ab|cd$ and internal edge of length $t$. Then

\begin{align}\label{CFtree}
CF_{abcd}=\left(1-\frac{2}{3}e^{-t},
\frac{1}{3}e^{-t},
\frac{1}{3}e^{-t}\right).
\end{align}

As mentioned in \cite{Solis-Lemus2016}, for unrooted species networks the concordance factors do not depend on the placement of the root in the species network, as long as the root is placed in a way consistent with the direction of the hybrid edges.  This fact is fully shown in Section \ref{Q CF sn}, as we explore quartet concordance factors more fully.

\begin{definition}
	Let $\mathcal{N}^+$ be a metric rooted level-1 network on  $X$. Given a set of distinct taxa $\{a,b,c,d\}$, we define the \textbf{ordering of $CF_{abcd}$ on $\mathcal{N}^+$} as the natural decreasing order of  $CF_{AB|CD},$ $CF_{AC|BD}$, $CF_{AD|BC}$ in the real line.\\
\end{definition}

For example if $t>0$ the ordering  of the concordance factors in equation \eqref{CFtree} is given by $$CF_{AB|CD}> CF_{AC|BD}=CF_{AD|BC}.$$ Many arguments towards the main result of this work use the ordering of $CF_{abcd}$, and not its precise values.

\section{Computing quartet concordance factors}\label{Q CF sn}
 
In this section we  show how to express the  concordance factors arising on a MRCA  quartet network as a linear combination of the concordance factors arising on quartet trees. This enables us to see how the ordering of concordance factors reflects the network topology, and how the precise root location does not matter. We fully address issues that are important when the MRCA  quartet network is induced from a larger one on more taxa; these are ommited in \cite{Solis-Lemus2016}. \\

Let $\mathcal{N}^+$ be a (metric or topological) rooted level-1 network on  $X$ and let $\{a,b,c,d\}$ be a set of distinct taxa of $X$.  Then the induced unrooted network on 4 taxa $\mathcal{Q}^-_{abcd}$  is a  (metric or topological) unrooted level-1 network.  By Proposition \ref{prop::ancestralinduced},   $\mathcal{Q}^-_{abcd}$ is the same graph as  $(\mathcal{N}^+_{abcd})^-$ and $(\mathcal{N}^\oplus_{abcd})^-$. Any cycle in  $\mathcal{N}^\oplus_{abcd}=\mathcal{Q}^\oplus_{abcd}$ induces a cycle in $\mathcal{Q}^-_{abcd}$. A cycle $C$ in  $\mathcal{Q}^\oplus_{abcd}$ of size $k$, induces a cycle in $\mathcal{Q}^-_{abcd}$ of either  size $k$ (when $C$ does not contain the MRCA($a,b,c,d$)) or size $k-1$ (otherwise).  For  convenience when we refer to the size of a cycle $C$ in $\mathcal{Q}^\oplus_{abcd}$ we mean the size of the induced cycle in $\mathcal{Q}^-_{abcd}$.

\begin{lemma}\label{lem:numberofcycles}
	Let  $\mathcal{Q}^-_{abcd}$ be a  metric unrooted level-1 unrooted quartet network. The   number of $k$-cycles  in $\mathcal{Q}^-_{abcd}$  is $0$ for $k\geq 5$, at most $1$ for $k=4$ in which case there is no $3$-cycle, and at most $2$ for $k=3$. 
\end{lemma}
\begin{proof}
	Suppose that $\mathcal{Q}^-_{abcd}$ has a cycle $C=C_v$ of size $k$.  Then there is an associated partition of taxa into $k$ $v$-blocks. Trivially none of these blocks can be empty, so $k\leq 4$. 
	
Suppose that there are two cycles, a cycle  $C_1$ of size $k_1$ and $C_2$ of size $k_2$ with $k_i\geq 3$, $i=1,2$. Since $\mathcal{Q}^-_{abcd}$ is level-1, by removing these two cycles we induce a partition of the taxa into at least $k_1+k_2-2$ blocks.  None of the blocks of this partition can be empty, so $k_1+k_2-2\leq 4$.   Hence there is a most one cycle of size $4$ or at most two cycles of size 3. Moreover there cannot be  a cycle of size $3$ and  a cycle of size $4$ in the same unrooted quartet network.
	
Suppose that there are three cycles, a cycle  $C_1$ of size $k_1$,  $C_2$ of size $k_2$, and $C_3$ of size $k_3$ with $k_i\geq 3$, $i=1,2,3$. By removing these three cycles we induce a partition of the taxa into at least $k_1+k_2+k_3-3$ blocks, so $k_1+k_2+k_3-3\leq 4$ which is a contradiction since $k_i\geq 3$.
 $\Box$\end{proof}

Our arguments will depend on the number of descendants on the hybrid node of a cycle, so we introduce additional terminology. An $n$-cycle with  exactly $k$ taxa descending from the hybrid node is referred to as a $n_k$-\textit{cycle}. Figure \ref{fig:23cycles} shows  the 6 different types of $2$-, $3$-, and $4$-cycles possible in an unrooted quartet network. 

\begin{figure}\begin{center}
	\includegraphics[scale=.2 ]{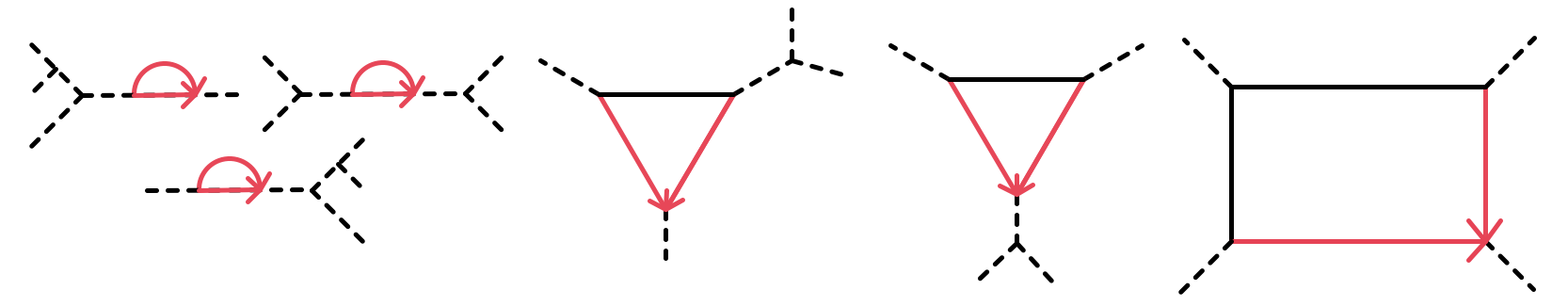}
	 \caption{(Left)  The three  types of $2$-cycles in an unrooted quartet network ($2_1$-,$2_2$- and a $2_3$-cycle); (Center)  The two  types of $3$-cycles in the unrooted quartet network ($3_1$- and a $3_2$-cycle). (Right) The only type of $4$-cycle in an unrooted quartet network (a $4_1$-cycle). The dashed lines represent subgraphs that may contain other cycles.}\label{fig:23cycles}
\end{center}\end{figure}

\begin{lemma}\label{lem::no2 3_2}Let $\mathcal{Q}^-_{abcd}$ be a  metric unrooted level-1 unrooted quartet network. Then  $\mathcal{Q}^-_{abcd}$  cannot have two $3_2$-cycles, or a $2_2$-cycle and a $4_1$-cycle.
\end{lemma}
\begin{proof}
	Suppose $Q=\mathcal{Q}^-_{abcd}$ has  two distinct $3_2$-cycles, $C_u$ and $C_v$.  Suppose $C_u$ has $u$-hybrid block $\{a,b\}$ and $u$-blocks $\{c\}$ and $\{d\}$. If we remove $C_u$ from $Q$, by the level-1 assumption $C_v$ is in one on the connected components. This implies that  2 of the 3 $v$-blocks must be contained in one of  $\{a,b\}$,  $\{c\}$ or  $\{d\}$. This is only possible if the $v$-hybrid block is $\{c,d\}$,   and the other $v$-blocks are $\{a\}$ and $\{b\}$. Thus $Q$ must be as the network  in Figure  \ref{fig::2 3c2}, where $u$ is below $v$ and $v$ is below $u$, contradicting that $Q$ is  induced from a rooted network.
	
	Now suppose that $Q$ has a 4-cycle  and a $2_2$-cycle. The $4$-cycle induces 4 singleton blocks. By the level-1 condition  at least one of the blocks induced by the $2_2$-cycle has to be contained in  a singleton block. That is impossible since the blocks induced by the $2_2$-cycle have size 2.   $\Box$\end{proof}

\begin{figure}\begin{center}
	\includegraphics[scale=.27]{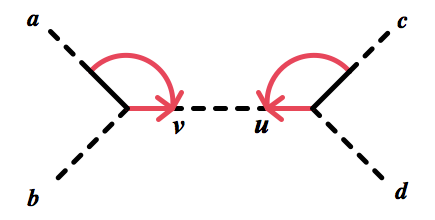}
	\caption{A graph with two $3_2$ cycles. Each dashed edge represents a chain of 2-cycles with, possibly, other cycles.}\label{fig::2 3c2}
\end{center}\end{figure}

Lemmas \ref{lem:numberofcycles} and \ref{lem::no2 3_2} determine all possible topological structures for unrooted quartet networks which are shown in Figure \ref{fig::n2s}.

\begin{figure}\begin{center}
	\includegraphics[scale=.25]{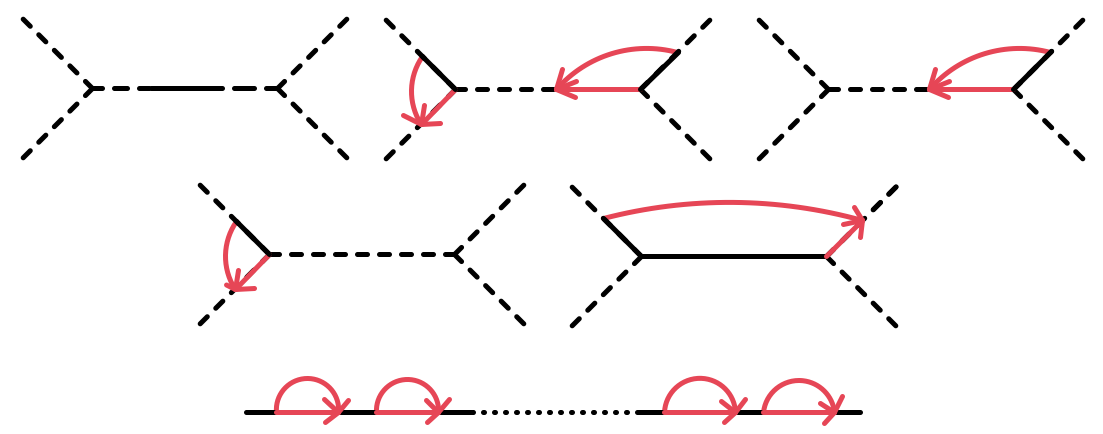}
	\caption{Possible structures for  unrooted quartet networks. Every dashed arrow represents a chain of an arbitrary number of 2-cycles, as the one in the bottom of the Figure.  The direction of these 2-cycles must be such that the obtained graph is induced from a rooted network. }\label{fig::n2s}
\end{center}\end{figure}

\subsection{Concordance factor formulas for   quartet networks}
 
Next we prove a number of ``reduction" lemmas relating concordance factors for quartet networks to those for networks with fewer cycles. This allows us to express the network concordance factors in terms of those for trees.   The following observation is useful through this section.

\begin{obs}\label{obs:quartetcoales}
	Given a rooted metric species  quartet network, under the NMSC model the first coalescent event determines the unrooted topology of a quartet gene tree.
\end{obs}


\begin{figure}\begin{center}
	\includegraphics[scale=.20]{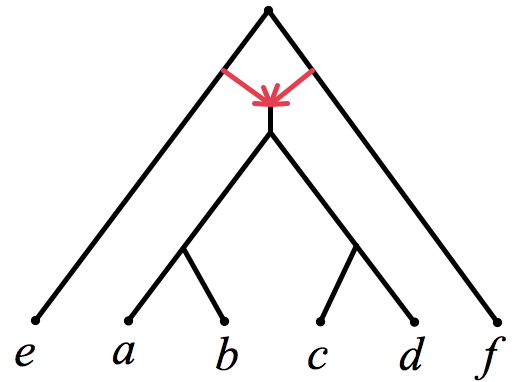}
	\caption{A  level-1 rooted  network  where the root differs from the $\text{MRCA}(a,b,c,d)$. }\label{fig::cycabove}
\end{center}\end{figure}
As illustrated in Figure \ref{fig::cycabove}, in passing from a rooted network on $X$ to a rooted induced network on $Z\subset X$, $\mathcal{N}^+_{Z}$, we may find there is a network structure above MRCA($Z$), a chain of $2$-cycles by  Proposition \ref{prop::structureabove}. \textit{A priori}, this could have an impact on the behavior of the NMSC model on $\mathcal{N}^+_Z$. For quartet concordance factors, however, this additional structure has no impact, and we effectively snip it off. Formally, we have the following.

\begin{theorem}\label{thm::MRCA}
	Let $\mathcal{N}^+$ be a level-1 rooted metric network on   $X$ and let $a,b,c,d$ be distinct taxa of $X$.
	Under the NMSC model,  $CF_{abcd}$  can be computed from the MRCA network $\mathcal{Q}^\oplus_{abcd}$.
\end{theorem}
\begin{proof}
		In any realization of the coalescent process if there are fewer than $4$ lineages at the MRCA($a,b,c,d$) in $\mathcal{N}^+_{abcd}=\mathcal{Q}^+_{abcd}$, then a coalescent event has occurred  below and therefore the unrooted gene tree topology has been determined. Thus we condition on $4$ lineages being present at MRCA($a,b,c,d$).

	There are 2 rooted shapes for 4-taxon gene trees, the caterpillar and balanced trees. Regardless of the ancestral chain of 2-cycles  above MRCA($a,b,c,d$), conditioned on one of these shapes,  exchangeability of lineages under the coalescent tells us all labeled versions of that specific shape will have equal probability.  While the rooted shapes might have different probability,  since there is only 1 unrooted  shape, all labellings of it must be equally probable.  This is the same as if there were no ancestral cycles.  	Therefore $CF_{abcd}(\mathcal{Q}^\oplus_{abcd})=CF_{abcd}(\mathcal{Q}^+_{abcd})$.   $\Box$\end{proof}

This argument can be modified to apply to  5 taxa, but not 6 or more, since then there is more than 1 unrooted shape.  \\

Let $Q^\oplus=\mathcal{Q}^\oplus_{abcd}$ be a level-1 MRCA quartet network and let $C_v$ be a cycle in $Q^\oplus$,   with   hybrid node $v$ and hybrid edges $h_1$ and $h_2$, where  $\gamma=\gamma_{h_1}$. The following notation is used throughout this section: \\
\begin{itemize}[label=$\bullet$]
	\item  $Q_1^\oplus$ denotes the rooted quartet network  obtained from $Q^\oplus$ by removing $h_2$.
	\item  $Q_2^\oplus$ denotes the rooted quartet network  obtained from $Q^\oplus$ by removing $h_1$.
	\item $Q_0^\oplus$ denotes the rooted quartet network obtained from $Q^\oplus$ by suppressing $C_v$; if the root of $Q^\oplus$ is in $C_v$, the node obtained by identification in the suppression process is the root of $Q_0^\oplus$.  \\
\end{itemize}
Note that  $Q_i^\oplus$, for $i=1,2$ have degree 2 nodes, and thus are not  binary. This does not affect the coalescent process in any way and by suppressing such nodes we obtain a binary MRCA network. In a slight abuse of notation, we use $Q_i^\oplus$ to denote both of these networks, as needed in our arguments.\\

In  arguments  on computing concordance factors we  often need to designate how many lineages are present at a hybrid node in a realization of the coalescent process. To handle this formally, given a rooted metric species network $\mathcal{N}^+$ on $X$,   we define  the random variable $K_v$ to be  the number of lineages at node $v$, where $K_v$ takes values in $\{1,...,l_v\}$, where $l_v$ is the number of taxa below $v$. We can extend this concept to hybrid nodes in $\mathcal{N}^-$, since a hybrid node in $\mathcal{N}^-$ induces an orientation of the nodes that are descending from it. 

 Let $Q^\oplus=\mathcal{Q}^\oplus_{abcd}$ be a level-1 MRCA quartet network and let $C_v$ be a cycle in $Q^\oplus$, with   hybrid node $v$, which induces a cycle  $C'_v$ in $\mathcal{Q}^-_{abcd}$. If  $C'_v$ has size $2$, then $1\leq l_v\leq 3$; if $C_v'$ has size three, then  $1\leq l_v\leq 2$; and if $C_v'$ has size four then $l_v=1$, as shown in Figure \ref{fig:23cycles}. \\
 
We first show that cycles in $\mathcal{Q}^\oplus_{abcd}$  that induce $2_1$-cycles or $2_3$-cycles in $\mathcal{Q}^-_{abcd}$  have no impact on concordance factors.

\begin{lemma}\label{lem::formula2:1}
	Let $Q^\oplus=\mathcal{Q}^\oplus_{abcd}$ be a metric level-1 MRCA quartet network and  let $C_v$ be a cycle in $Q^\oplus$ that induces a $2_1$-cycle in $\mathcal{Q}^-_{abcd}$. Then $CF(Q^\oplus)=CF(Q_0^\oplus)$.
\end{lemma}
\begin{proof}
	Let $K=K_v$. Since $C_v$ induces a $2_1$-cycle in $\mathcal{Q}^-_{abcd}$, $P(K=1)=1$.
	Then
	\begin{align*}
	CF(Q^\oplus)&=P(K=1)CF(Q^\oplus\mid K=1)\\
	&=P(K=1)[\gamma CF(Q_1^\oplus\mid K=1) +(1-\gamma)CF(Q_2^\oplus\mid K=1)]\\
	&=\gamma CF(Q_1^\oplus)+(1-\gamma)CF(Q_2^\oplus)
	\end{align*}
	
	If the root of $Q^\oplus$ is not in $C_v$,  no lineages can coalesce on the edges that differ in  $Q_1^\oplus$ and $Q_2^\oplus$. Thus,  $$CF(Q_1^\oplus)=CF(Q_2^\oplus)=CF(Q_0^\oplus),$$
	and the claim is established in this case.
	
	Now suppose the root $r$ of $Q^\oplus$ is in $C_v$, and $C_v$ has nodes $r$, $u$, $v$, and edges $(r,v)$, $(r,u)$, $(u,v)$. Without loss of generality suppose that the taxon below $v$  is $d$.  Since $u$ is a tree node it has another descendant $y$. Note that $Q_1^\oplus$ and $Q_2^\oplus$ have the same topology, moreover, they just differ in the edge length from the root to $y$. Define a random variable $K'$, by $K'=1$ if  there has been a coalescent event before $a$, $b$, and $c$ trace back to $y$ and $K'=0$ otherwise. If $K'=1$, the unrooted topology has been determined and thus $$CF(Q_1^\oplus\mid K'=1)=CF(Q_2^\oplus\mid K'=1)=CF(Q_0^\oplus\mid K'=1).$$
	Also,  by  Proposition 11  in \cite{Allman2011},
	 $$CF(Q_1^\oplus\mid K'=0)=CF(Q_2^\oplus\mid K'=0)=CF(Q_0^\oplus\mid K'=0).$$
	  Thus $CF(Q^\oplus)=CF(Q_0^\oplus).$   $\Box$\end{proof}
\begin{lemma}\label{lem::formula2:3} 
	Let $Q^\oplus=\mathcal{Q}^\oplus_{abcd}$ be a  metric level-1 MRCA quartet network and let $C_v$ be a cycle in $Q^\oplus$, that induces a $2_3$-cycle in $\mathcal{Q}^-_{abcd}$. Then $CF(Q^\oplus)=CF(Q_0^\oplus)$.
\end{lemma}
\begin{proof}
	Let $K=K_v$, so  $K$ takes values in $\{1,2,3\}$. 
	Therefore 
	\begin{align}
	CF(Q^\oplus)&= P(K=1)CF(Q^\oplus\mid K=1)+P(K=2)CF(Q^\oplus\mid K=2)\notag\\
	&\quad+P(K=3)CF(Q^\oplus\mid K=3).\label{eq::QK}
	\end{align}
 
	If $K=1$ or 2 then at least one coalescent event has occurred,  so the unrooted gene tree topology  is already determined, and 
	\begin{align*}
	CF(Q^\oplus\mid K=k)=CF(Q_0^\oplus\mid K=k)\text{ for $k=1,2$.}
	\end{align*}

	The case $K=3$ requires more argument. Without loss of generality suppose that the three taxa descending from $v$ are $a$, $b$, and $c$. Denote by $\mathfrak{D}$ the random variable defined by $\mathfrak{D}=1$ if the lineage $d$ is involved in the first coalescent event and $\mathfrak{D}=0$ otherwise. Thus
	\begin{align}\label{eq::2_3mrca}
	CF(Q^\oplus\mid K=3) & = P(\mathfrak{D}=1)CF(Q^\oplus\mid K=3,\mathfrak{D}=1) \notag \\
	&\quad+P(\mathfrak{D}=0)CF(Q^\oplus\mid K=3,\mathfrak{D}=0). 
	\end{align}
	
	If $d$ is in the first coalescent event, by the exchangeability property of the NMSC, $a,b$ or $c$ are equally likely to be the other lineage involved in that event.This is the same as if the cycle was suppressed, so
	$$CF(Q^\oplus\mid K=3,\mathfrak{D}=1)=\left(\frac{1}{3},\frac{1}{3},\frac{1}{3}\right)=CF(Q_0^\oplus\mid K=3,\mathfrak{D}=1)$$
	If $d$ is not in the first coalescent event, this event involves  only two of $a,b,$ and $c$, with each pair  equally likely by exchangeability. 
	 This is also the same as if the cycle was suppressed, so
	$$CF(Q^\oplus\mid K=3,\mathfrak{D}=0)=\left(\frac{1}{3},\frac{1}{3},\frac{1}{3}\right)=CF(Q_0^\oplus\mid K=3,\mathfrak{D}=0)$$

 	Thus  by  equations \eqref{eq::QK} and \eqref{eq::2_3mrca},  $CF(Q^\oplus)=CF(Q_0^\oplus)$.  $\Box$\end{proof}

Together, the preceding Lemmas yield the following.
\begin{corollary}\label{cor:2isgood}
	Let $Q^\oplus=\mathcal{Q}^\oplus_{abcd}$ be a  metric level-1 MRCA quartet network and let $\widetilde{Q}^\oplus$ be the MRCA network obtained by suppressing all cycles that induce either  $2_3$- or a $2_1$-cycles in $\mathcal{Q}^-_{abcd}$. Then $CF(Q^\oplus)=CF(\widetilde{Q}^\oplus)$.
\end{corollary}

While $2_1$- and $2_3$-cycles have no impact on concordance factors, things are not quite so simple for other types of cycles. 
 
\begin{lemma}\label{lem::formula2:2}Let $Q^\oplus=\mathcal{Q}^\oplus_{abcd}$ be a  metric level-1 MRCA quartet network and let $C_v$ be a cycle in $Q^\oplus$, that induces a $2_2$-cycle in $\mathcal{Q}^-_{abcd}$. Then $$CF(Q^\oplus)=\gamma^2CF(Q_1^\oplus)+(1-\gamma)^2CF(Q_2^\oplus)+2\gamma(1-\gamma)CF(Q_0^\oplus).$$
\end{lemma}
\begin{proof}
	Let $K=K_v$ with values in $\{1,2\}$, so that
	\begin{align*}
	CF(Q^\oplus)=P(K=1)CF(Q^\oplus\mid K=1)+P(K=2)CF(Q^\oplus\mid K=2).
	\end{align*}

 Suppose the root $r$ of $Q^\oplus$ is not in $C_v$, so $C_v$ is also a $2_2$-cycle in $Q^\oplus$. Note that 
	\begin{align*} 
	CF(Q^\oplus\mid K=2)&=\gamma^2CF(Q_1^\oplus\mid K=2)+(1-\gamma)^2CF(Q_2^\oplus\mid K=2)\\
	&\quad +2\gamma(1-\gamma)CF(Q_0^\oplus\mid K=2).
	\end{align*}
 Thus we will express  $CF(Q^\oplus\mid K=1)$ in a similar fashion.  If $K=1$  the gene tree topology has been determined before the lineages enter $v$. Thus $CF(Q_i^\oplus\mid K=1)=CF(Q\mid K=1)$  for $i\in\{0,1,2\}$ and
	\begin{align}\label{cfeq1}
	CF(Q^\oplus\mid K=1)&=\gamma^2CF(Q_1^\oplus\mid K=1)+(1-\gamma)^2CF(Q_2^\oplus\mid K=1) \notag\\
	&\quad+2\gamma(1-\gamma)CF(Q_0^\oplus\mid K=1);
	\end{align}
	  by summing the result holds when $r$ is not in $C_v$. 
	
	Now suppose that $r$ is in  $C_v$, and $C_v$ has nodes $r$, $v$, $u$. Without loss of generality suppose that the taxa below $v$ are $c$ and $d$.  Since $u$ is a tree node it has another descendant $y$. Define a random variable $K_y$  to be the number of lineages at $y$.  Note that $K$ and $K_y$ are independent, with values in $\{1,2\}$. If  either $K$ or $K_y$ is 1, one coalescent event has occurred and the unrooted gene tree topology has been determined so $CF(Q_i^\oplus\mid K=1\text{ or }K_y=1)$ are equal for $i\in\{0,1,2\}$, and
	\begin{align} \label{eq::2_2mrca}
	CF(Q^\oplus\mid K=1\text{ or }K_y=1)&=\gamma^2CF(Q_1^\oplus\mid K=1\text{ or }K_y=1)\notag\\
	&\quad+(1-\gamma)^2CF(Q_2^\oplus\mid K=1\text{ or }K_y=1)\notag\\
	&\quad+2\gamma(1-\gamma)CF(Q_0^\oplus\mid K=1\text{ or }K_y=1).
	\end{align}
	Now suppose that $K$ and $K_y$ are both 2. Let $T_c$ and $T_d$ be the trees shown on Figure \ref{fig:TdTc}.
	Therefore
	\begin{align*}
	CF(Q^\oplus\mid K=2,K_y=2)&=\gamma^2CF(Q_1^\oplus\mid K=2,K_y=2)\\
	&\quad+(1-\gamma)^2CF(Q_2^\oplus\mid K=2,K_y=2)\\
	&\quad+\gamma(1-\gamma)CF(T_c\mid K_y=2)\\
	&\quad+\gamma(1-\gamma)CF(T_d\mid K_y=2). 
	\end{align*}
	
	By Proposition 3  in \cite{Allman2011}, $CF(T_d\mid K_y=2)=CF(T_c\mid K_y=2)$, and in fact they equal $CF(Q_0^\oplus\mid K=2,K_y=2)$. This is because in $Q_0^\oplus$ the suppression of the cycle identifies the nodes $r$, $u$, and $v$, so conditioned on $K=2$, $K_y=2$ we may view the coalescent process on $Q_0^\oplus$ as that in the $4$-taxon tree $((a,b):l,(c,d):0)$ where $l$ is the length of $(u,y)$. By Proposition 11  in \cite{Allman2011}, $CF(T_c\mid K_y=2)=CF(Q_0^\oplus\mid K=2,K_y=2)$.  Therefore  
		\begin{multline*}
	CF(Q^\oplus\mid K=2,K_y=2)=\gamma^2CF(Q_1^\oplus\mid K=2,K_y=2)\\
	+(1-\gamma)^2CF(Q_2^\oplus\mid K=2,K_y=2) 
	2\gamma(1-\gamma)CF(Q_0^\oplus\mid K=2,K_y=2).
	\end{multline*}
	
	This together with equation \eqref{eq::2_2mrca} implies the claim.  $\Box$\end{proof}

\begin{figure}\begin{center}	 
	\includegraphics[scale=.22]{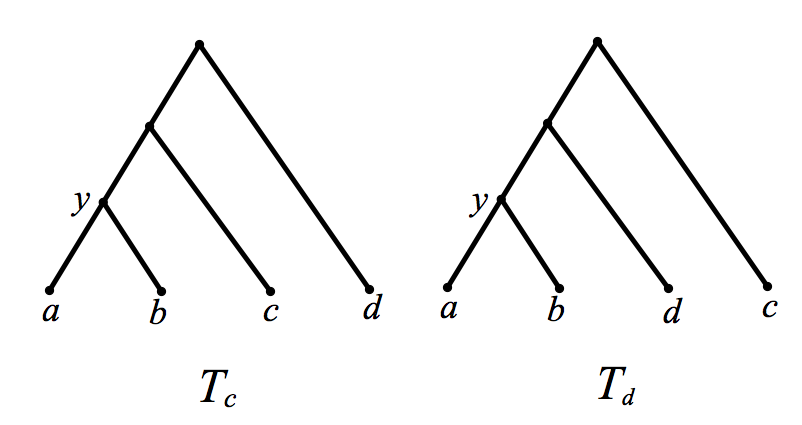} \caption{ The two trees $T_d$ and $T_c$ in the proof of Lemma \ref{lem::formula2:2}, obtained when $K=2$, $K_y=2$ and the lineages $c$ and $d$ trace  different hybrid edges.}\label{fig:TdTc}
\end{center}\end{figure}

\begin{lemma}\label{lem::formula43:1}Let $Q^\oplus=\mathcal{Q}^\oplus_{abcd}$ be a metric  level-1 MRCA quartet network and let $C_v$ be a cycle in $Q^\oplus$, that induces either  a 4-cycle or a $3_1$-cycle in $\mathcal{Q}^-_{abcd}$.  Then $$CF(Q^\oplus)=\gamma CF(Q_1^\oplus)+(1-\gamma)CF(Q_2^\oplus).$$
\end{lemma}
\begin{proof}
	Letting $K=K_v$, then $P(K=1)=1$. Thus,
	\begin{align*}
	CF(Q^\oplus)&=P(K=1)CF(Q^\oplus\mid K=1)\\
	&=P(K=1)(\gamma CF(Q_1^\oplus\mid K=1)+(1-\gamma)CF(Q_2^\oplus\mid K=1))\\
	&=\gamma CF(Q_1^\oplus)+(1-\gamma)CF(Q_2^\oplus).
	\end{align*}
 $\Box$\end{proof}

 It remains to consider  a $3_2$-cycle. For this case it helps to introduce new terminology. 	Let $G$ be a semidirected graph and $v$ be a node in $G$ with indegree 2 and outdegree 0. Let $h_v$ and $h_v'$ be the edges incident to $v$ and let $u$ and $u'$ the parent nodes in $h_v$ and $h_v'$ respectively. We refer to \textit{disjointing} $h_v$ and $h_v'$ from $v$ as the process of 1) deleting $v$ from $G$; 2) introducing nodes $w$ and $w'$; 3) introducing directed edges $(u,w)$ and $(u',w')$.

Let $Q^\oplus=\mathcal{Q}^\oplus_{abcd}$ be a  metric level-1 MRCA quartet network, and $C_v$ a cycle in $Q^\oplus$, that induces a $3_2$-cycle in $\mathcal{Q}^-_{abcd}$. Without loss of generality suppose that $a$ and $b$ are the taxa below $v$.   Let $Q_a^\oplus$ be the network  obtained from $Q^\oplus$ by 1) deleting  everything below $v$; 2) disjointing $h_1$ and $h_2$ from $v$; 3) labeling a leaf that is currently unlabeled  by $a$ and the other  unlabeled leaf by $b$. We construct $Q_b^\oplus$ by swapping the labels $a$ and $b$ in $Q_a^\oplus$. Figure \ref{fig::disjoint} depicts an particular example of this.

\begin{figure}\begin{center}	 
	\includegraphics[scale=.18]{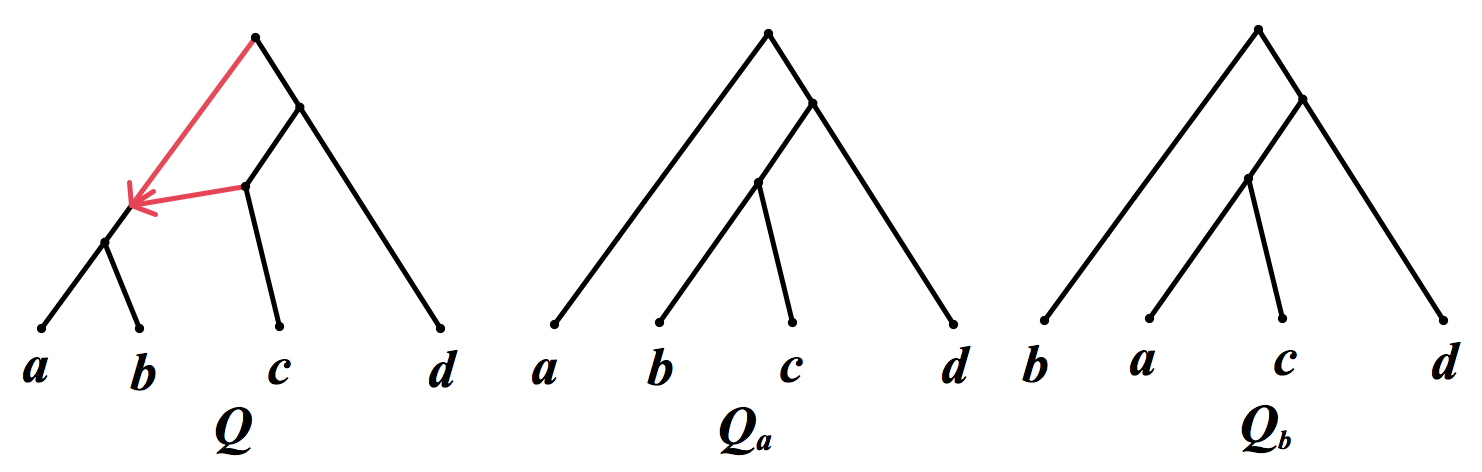} \caption{ A MRCA quartet $Q^\oplus$ with a cycle $C$ that induces a $3_2$-cycle in the unrooted quartet and  the   graphs obtained by deleting everything below the hybrid node,  disjointing, and labeling the leaves.}\label{fig::disjoint}
\end{center}\end{figure}

\begin{lemma}\label{lem::formula3:2}
	Let $Q^\oplus=\mathcal{Q}^\oplus_{abcd}$ be a  metric level-1 MRCA quartet network,  $C_v$ be a cycle in $Q^\oplus$, that induces a $3_2$-cycle in $\mathcal{Q}^-_{abcd}$ and let $K=K_v$. Suppose that the two taxa below $v$ are $a$ and $b$, then
		\begin{align*}
		CF(Q^\oplus)&=\gamma^2CF(Q_1^\oplus)+(1-\gamma)^2CF(Q_2^\oplus)\\
		&\quad+P(K=1)2\gamma(1-\gamma)CF(Q_0^\oplus\mid K=1)\\
		&\quad +P(K=2)\gamma(1-\gamma) [CF(Q_a^\oplus)+CF(Q_b^\oplus)].
		\end{align*}
\end{lemma}
\begin{proof}
	By hypothesis $K$ takes values in $\{1,2\}$ and 
	$$CF(Q^\oplus)=P(K=1)CF(Q^\oplus\mid K=1)+P(K=2)CF(Q^\oplus\mid K=2).$$
	If $K=1$ the unrooted tree topology has been determined and  $CF(Q^\oplus\mid K=1) $ is given by the expression in equation \eqref{cfeq1}.  If $K=2$,
	\begin{align*}
	CF(Q^\oplus\mid K=2)&=\gamma^2 CF(Q_1^\oplus\mid K=2)+(1-\gamma)^2CF(Q_2^\oplus\mid K=2)\\
	&\quad + \gamma(1-\gamma) CF(Q_a^\oplus)+\gamma(1-\gamma)CF(Q_b^\oplus).
	\end{align*}
	Therefore,
	\begin{align*}
	CF(Q^\oplus)&=P(K=1)(\gamma^2CF(Q_1^\oplus\mid K=1)+(1-\gamma)^2CF(Q_2^\oplus\mid K=1)\\
	&\quad+2\gamma(1-\gamma)CF(Q_0^\oplus\mid K=1)\\
	&\quad+P(K=2)[\gamma^2CF(Q_1^\oplus\mid K=2)+(1-\gamma)^2CF(Q_2^\oplus\mid K=2)\\
	&\quad+\gamma(1-\gamma) CF(Q_a^\oplus)+\gamma(1-\gamma)CF(Q_b^\oplus)],
	\end{align*}
	which yields the claim.
 $\Box$\end{proof}


These  Lemmas together imply that concordance factor for rooted quartet networks actually depend only on the unrooted network. This is formalized in the following.

\begin{proposition}\label{prop::CFunrooted}
	Let $Q=\mathcal{Q}^\oplus_{abcd}$ and $\widetilde{Q}=\widetilde{\mathcal{Q}}^\oplus_{abcd}$ be metric  level-1 MRCA quartet networks which induce the same unrooted network $\mathcal{Q}^-_{abcd}=\tilde{\mathcal{Q}}^-_{abcd}$. Then $CF(Q)=CF(\widetilde{Q})$.
\end{proposition}
\begin{proof}
	We prove this by induction on    the number of cycles in   $\mathcal{Q}^-_{abcd}$. When  there are no cycles in   $\mathcal{Q}^-_{abcd}$,  $Q$ and $\widetilde{Q}$ are trees, and by  Proposition 3  in \cite{Allman2011},  $CF(Q)=CF(\widetilde{Q})$. Assume now the result is true when there are fewer than $k+1$ cycles and  that  $\mathcal{Q}^-_{abcd}$ has $k+1$ cycles. Let $C_v$ be a cycle in $\mathcal{Q}^-_{abcd}$ with hybrid edges $h_1$ and $h_2$, by Lemmas \ref{lem::formula2:1}, \ref{lem::formula2:3}, \ref{lem::formula2:2}, \ref{lem::formula43:1}, and \ref{lem::formula3:2}, we can express the concordance factors of $Q$ and $\widetilde{Q}$ in terms of  networks with one fewer cycle. 
	Note that these networks for $Q$ and $\widetilde{Q}$ have the same   unrooted metric structure. Thus by  the induction hypothesis $CF(\widetilde{Q}_i)=CF(Q_i)$, for $i=0,1,2$, and therefore $CF(\widetilde{Q})=CF(Q)$.  $\Box$\end{proof}	

\begin{corollary}\label{cor::quartetisgood}
	Let $\mathcal{N}^+$ be a level-1 rooted metric network on   $X$ and let $a,b,c,d$ be distinct taxa of $X$.
	Under the NMSC,  $CF_{abcd}=CF(\mathcal{Q}^\oplus_{abcd})$  can be computed from the unrooted network $\mathcal{Q}^-_{abcd}$. 
\end{corollary}

We  indicate how to compute the concordance factors of a MRCA network $\mathcal{Q}_{abcd}^\oplus$ from the unrooted quartet network $Q=\mathcal{Q}^-_{abcd}$ without having to introduce a root.  For $Q=\mathcal{Q}^-_{abcd}$  a unrooted metric level-1 quartet network,   where using Corollary  \ref{cor::quartetisgood} we define $CF(Q)=CF(\mathcal{Q}_{abcd}^\oplus)$ :

\begin{enumerate}
	\item[i)]  Let $Q'$ be the graph obtained from $Q$ by suppressing all $2_3$- and $2_1$- cycles. By Corollary \ref{cor:2isgood}, $CF(Q)=CF(Q')$. If $Q$ has a $4$-cycle go to step (ii), otherwise go to step (iii).
	\item[ii)]  By Lemma \ref{lem:numberofcycles} and Lemma \ref{lem::no2 3_2} there are no $3_1$-, $3_2$- or $2_2$-cycles in $Q$, and thus none in $Q'$. Then $Q'$ only has a $4$-cycle so apply Lemma \ref{lem::formula43:1} to $Q'$. Since $Q'_1$ and $Q'_2$ are quartet trees,  use the formula in equation \eqref{CFtree}. 
	\item[iii)]  There are at most two $3_1$-cycles in $Q'$. Choose one arbitrarily and apply Lemma \ref{lem::formula43:1}.   If $Q'_1$ and $Q'_2$ still have a $3_1$-cycle,  apply Lemma \ref{lem::formula43:1} again to $Q'_1$ and $Q'_2$.  
	\item[iv)] We have now expressed    concordance factors of $Q$ in terms of concordance factors of unrooted quartet networks with no $2_1$-,$2_3$-,$3_1$, or $4$-cycles. Apply Lemma \ref{lem::formula2:2}  to these networks, by for instance  choosing a $2_2$-cycle  with smallest graph theoretical distance from its hybrid node to a leaf, repeating until no 2-cycle remains.
	\item[v)] We have now an expression of the concordance factors of $Q$ in terms of concordance factors of unrooted  quartet networks with at most one $3_2$-cycle. Apply Lemma \ref{lem::formula3:2}.  Then we have removed all cycles, and   the concordance factors are now in terms of unrooted quartet trees. The formula of equation \eqref{CFtree} completes the calculation. \\
\end{enumerate}


The use of these Lemmas and Theorem is illustrated by a few examples.

\begin{example}\label{ex:2_2}
	Consider the unrooted quartet network  shown in Figure \ref{fig::2c2t}. By  Lemma \ref{lem::formula2:2}, with $x_i=e^{-t_i}$, the quartet concordance factors are given by:
	\begin{figure}\begin{center}
		\includegraphics[scale=.3]{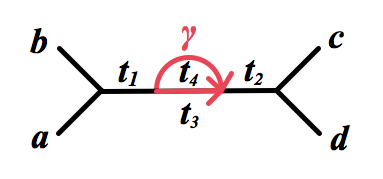}\\
		\caption{An unrooted quartet with a single $2_2$-cycle.}\label{fig::2c2t}
	\end{center}\end{figure}
	\begin{align}
	CF_{AB|CD}&=\left(1-\gamma\right)^2\left(1-\frac{2}{3}x_1x_2x_3\right)+2\gamma\left(1-\gamma\right)\left(1-\frac{2}{3}x_1x_2\right)\notag\\
	&\quad+\gamma^2\left(1-\frac{2}{3}x_1x_2x_4\right), \notag\\
	CF_{AC|BD}&=CF_{AD|BC}  \label{eq::2c2}\\
	&=\left(1-\gamma\right)^2\left(\frac{1}{3}x_1x_2x_3\right)+2\gamma\left(1-\gamma\right)\left(\frac{1}{3}x_1x_2\right)+\gamma^2\left(\frac{1}{3}x_1x_2x_4\right). \notag
	\end{align}
\end{example}

\begin{example}\label{ex:3_1}
	Consider the unrooted quartet network shown in Figure \ref{fig::3c1t}. By  Lemma \ref{lem::formula43:1}, with $x_i=e^{-t_i}$, the quartet concordance factors are given by:
	\begin{figure}\begin{center}
		\includegraphics[scale=.3]{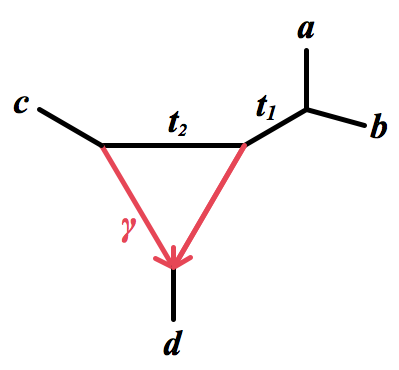}\\
		\caption{An unrooted quartet with a single $3_1$-cycle.}\label{fig::3c1t}
	\end{center}\end{figure}
	\begin{align}
	CF_{AB|CD}&= \left(1-\gamma\right)\left(1-\frac{2}{3}x_1\right)+\gamma\left(1-\frac{2}{3}x_1x_2\right)  ,\notag\\
	CF_{AC|BD}&= CF_{AD|BC}= \left(1-\gamma\right)\left(\frac{1}{3}x_1\right)+\gamma\left(\frac{1}{3}x_1x_2\right). \label{eq::3c1} 
	\end{align}
\end{example}

\begin{example}\label{ex:4_1}
	Consider the unrooted quartet  network shown in Figure \ref{fig::4c1t}. By Lemma \ref{lem::formula43:1}, with $x_i=e^{-t_i}$, the quartet concordance factors  are given by:
	\begin{figure}\begin{center}
		\includegraphics[scale=.25]{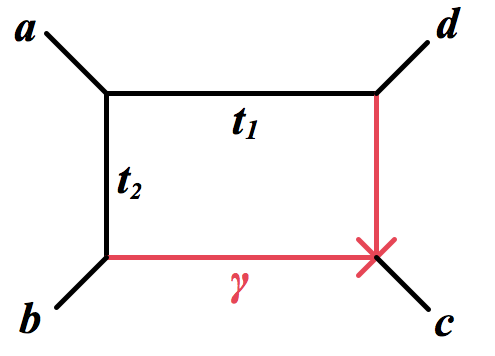}\\
		\caption{An unrooted quartet with a single $4_1$-cycle.}\label{fig::4c1t}
	\end{center}\end{figure}
	\begin{align}
	CF_{AB|CD}=& \left(1-\gamma\right)\left(1-\frac{2}{3}x_1\right)+\gamma\left(\frac{1}{3}x_2\right),\notag\\
	CF_{AC|BD}=& \left(1-\gamma\right)\left(\frac{1}{3}x_1\right)+\gamma\left(\frac{1}{3}x_2\right),\label{eq::4cyc}\\
	CF_{AD|BC}=& \left(1-\gamma\right)\left(\frac{1}{3}x_1\right)+\gamma\left(1-\frac{2}{3}x_2\right)   .\notag
	\end{align}	
\end{example}

\begin{example}\label{ex:3_2}
	Consider the unrooted quartet network shown in Figure \ref{fig::3c2t}. Note that $CF(Q_0\mid K=1)=(1,0,0)$. 
	By  Lemma \ref{lem::formula3:2}, with $x_i=e^{-t_i}$, the quartet concordance factors are given by:
	\begin{figure}\begin{center}
		\includegraphics[scale=.3]{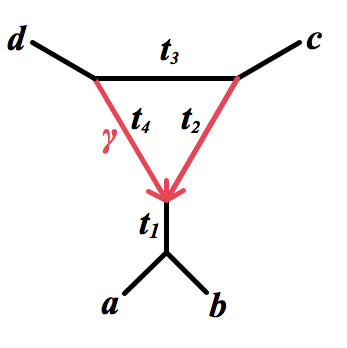}\\
		\caption{An unrooted quartet with a single $3_2$-cycle.}\label{fig::3c2t}
	\end{center}\end{figure}
	\begin{align}
	CF_{AB|CD}&=\left(1-\gamma\right)^2\left(1-\frac{2}{3}x_1x_2\right)+2\gamma\left(1-\gamma\right)\left(1-x_1+\frac{1}{3}x_1x_3\right)\notag\\
	&\quad +\gamma^2\left(1-\frac{2}{3}x_1x_4\right),\notag\\
	CF_{AC|BD}&=CF_{AD|BC}\label{eq::3cyc}\\
	&=\left(1-\gamma\right)^2\left(\frac{1}{3}x_1x_2\right)+\gamma\left(1-\gamma\right)x_1\left(1-\frac{1}{3}x_3\right)+\gamma^2\left(\frac{1}{3}x_1x_4\right).\notag
	\end{align}
\end{example}

\section{The Cycle property}\label{C section}

In this section we focus  on   the ordering  by magnitude of the concordance factors. 
\begin{proposition}\label{prop:Cfspreserv}Let $Q=\mathcal{Q}^-_{abcd}$ be a  metric unrooted level-1 quartet network with no $3_2$-cycle.
	The ordering of $CF_{abcd}(Q)$  is  the ordering of $CF_{abcd}(Q')$  where $Q'$ is obtained from $Q$ by suppressing all $2$-cycles and all $3_1$-cycles.
\end{proposition}
\begin{proof}
	By Corollary \ref{cor:2isgood}, $CF(Q)=CF(Q^*)$, where $Q^*$ is obtained from $Q$ by suppressing all $2_1$- and $2_3$-cycles. Therefore  we can assume $Q$ has no $2_1$- or  $2_3$-cycles. If $Q$ has a 4-cycle, it has no $3_1$- and no $2_2$-cycles and the claim is established.

	So suppose $Q$ has only $2_2$-cycles and $3_1$-cycles.  We proceed  by induction in the number of cycles, with the base case of 0 cycles trivial. 
	Assume the result  is true for unrooted quartet networks with $k$ $3_1$- and $2_2$-cycles and suppose $Q$ has $k+1$. Picking one cycle and applying one of Lemmas   \ref{lem::formula2:2} or \ref{lem::formula43:1} to $Q$, we can express the concordance factors of $Q$  as a convex combination of $CF(Q_0)$, $CF(Q_1)$ and $CF(Q_2)$. Note that $Q_0$, $Q_1$ and $Q_2$  have the same topology and by induction hypothesis,  $CF(Q_0)$, $CF(Q_1)$ and $CF(Q_2)$ have the same ordering as the concordance factors of  $Q'_0$, $Q'_1$ and $Q'_2$ respectively, the networks obtained after suppressing all  $2_2$- and $3_1$-cycles from  $Q_0$, $Q_1$ and $Q_2$. Since $Q'_0$, $Q'_1$, $Q_2$ and $Q'$ are trees with the same topology,  their concordance factors have the same ordering by equations \eqref{CFtree}. Thus $CF(Q_0)$, $CF(Q_1)$ and $CF(Q_2)$ have the same ordering, and ergo so does $CF(Q)$.  $\Box$\end{proof}

One consequence of Proposition \ref{prop:Cfspreserv} is that  for any  unrooted metric level-1 quartet network  $Q$ without a $3_2$- or a $4$-cycle, the ordering of the concordance factors is the same as the ordering of the concordance factors of a quartet tree. That is, the  two smallest elements of the concordance factors are equal. When this happens  we say that $Q$ is \textit{treelike}, since we could use equations \eqref{CFtree} to find a quartet tree with appropriate edge lengths and concordance factors equal to $CF(Q)$. However, not all unrooted quartet networks are treelike.  

\begin{example}\label{ex::cpbc}
	Let $\mathcal{Q}^-_{abcd}$ be the unrooted $3_2$-cycle quartet in Figure \ref{fig::3c2t}, where $\gamma=\frac{1}{2}$, $t_1=-\log\left(\frac{6}{7}\right)$, $t_2=-\log\left(\frac{6}{7}\right)$, $t_3=-\log\left(\frac{1}{14}\right)$ and $t_4=-\log\left(\frac{13}{14}\right)$.  By  the equations  in \eqref{eq::3cyc} we observe that  the concordance factors   are:
	\begin{align*}
	CF_{AB|CD}= \frac{32}{98},\text{ } CF_{AC|BD}=   \frac{33}{98},\text{ } CF_{AD|BC}=  \frac{33}{98}.
	\end{align*}
\end{example}

This motivates the following definition.

\begin{definition}\label{def::Cp}
	Let $\mathcal{N}^+$ be a metric rooted level-1  network on  $X$.
	We say that a set of four distinct taxa $s=\{a,b,c,d\}$ satisfies the \textbf{Cycle property} if  $\mathcal{Q}^-_s$ is not treelike, that is, if the two smallest values of $CF_s=CF(\mathcal{Q}^-_s)$ are not equal.
\end{definition}

The Cycle property is best understood geometrically. Denote by $\Delta_{2}$ the $2$-dimensional probability simplex,  the set of points in $ \mathbb{R}^3$ with nonnegative entries adding to 1.  
Observe that $CF_{abcd}\in\Delta_{2}$ for any distinct taxa $a,b,c,d$.   Figure \ref{fig:simplex1CP3} (left)   depicts the simplex where the black lines are the  points where the Cycle property is not satisfied; that is, the treelike unrooted quartet networks are those with concordance  factors   $(x,y,z)$ satisfying $x>\frac{1}{3}$, $y=z$ or  $y>\frac{1}{3}$, $x=z$ or  $z>\frac{1}{3}$, $x=y$. All points off these segments satisfy the Cycle property. \\
\begin{figure}\begin{center}
\includegraphics[scale=.27]{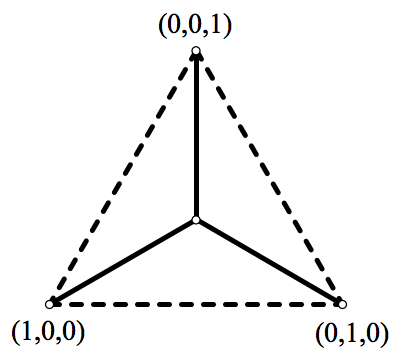}\includegraphics[scale=.27]{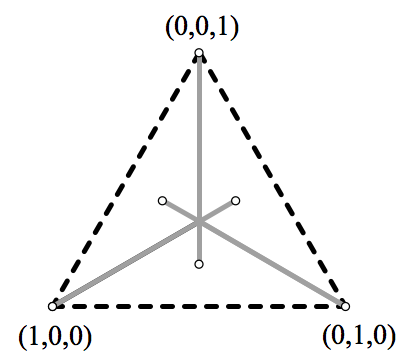}\includegraphics[scale=.27]{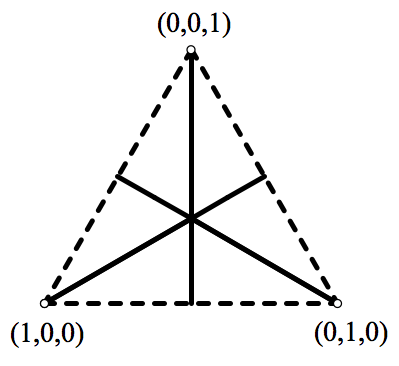} \caption{On the left a planar projection of the simplex $\Delta_2$, where the black lines represent concordance factors that are treelike. In the center, the  gray segments in $\Delta_2$ represent all the  concordance factors arising from unrooted quartet networks with a $3_2$-cycle. On the right,  the black lines represent the variety $V((x-z)(y-z)(x-y),x+y+z-1)$, these are  all concordance factors not satisfying the $BC$ property. }\label{fig:simplex1CP3}
\end{center}\end{figure}


\begin{proposition}\label{prop:3c2}Let $Q=\mathcal{Q}^-_{abcd}$ be a  metric unrooted level-1 quartet network with a $3_2$-cycle. Then $CF(Q)$  lies in 
	the  set $I$ defined by $x>\frac{1}{6}$, $y=z$ or  $y>\frac{1}{6}$, $x=z$ or  $z>\frac{1}{6}$, $x=y$, shown on the middle of  Figure \ref{fig:simplex1CP3}. Furthermore, for any point $(x,y,z)$ in this set there is such a $Q$ with $(x,y,z)=CF(Q)$.
\end{proposition}
\begin{proof}
	Let $s=\{a,b,c,d\}$ be a set of four distinct taxa and suppose that $\mathcal{Q}^-_{s}$  contains only a $3_2$-cycle,  as in Figure \ref{fig::3c2t}. Then $CF(\mathcal{Q}^-_{s})$ is given by the equations  \eqref{eq::3cyc} with $x_i=e^{-t_i}$, and in particular $CF_{AC|BD}=CF_{AD|BC}$.  To maximize $CF_{AD|BC}$ in  \eqref{eq::3cyc}, let $t_i\to  0$ for $i\in\{1,2,4\}$ and $t_3\to \infty$ to  obtain  a quadratic polynomial in $\gamma$,
	$$CF_{AD|BC}\to \frac{1}{3}(1-\gamma)^2+\gamma(1-\gamma)+\frac{1}{3}\gamma^2,$$
	
	whose maximum value is $\frac{5}{12}$ and it is attained at $\gamma=\frac{1}{2}$. For these values, we obtain $CF(\mathcal{Q}^-_{s})\to\left(\frac{2}{12},\frac{5}{12},\frac{5}{12}\right)$. To minimize $CF_{AD|BC}$ it is enough to  let $t_1\to\infty$, so $CF(\mathcal{Q}^-_{s})\to\left(1,0,0\right)$. 
	
	Let $\mathcal{L}$ be the open line segment with endpoints $\left(1,0,0\right)$ and $\left(\frac{2}{12},\frac{5}{12},\frac{5}{12}\right)$. Since $CF(\mathcal{Q}^-_{s})$ is continuous in $t_i$ and $\gamma$, its image is a connected set on the line $(x,y,y)$ containing points arbitrarily close to the endpoints of $\mathcal{L}$. Thus the image of $CF(\mathcal{Q}^-_{s})$ is $\mathcal{L}$. Permuting taxon names shows every point in the set $I$ is a concordance factor for a network with a $3_2$-cycle.
	
	Now suppose $\mathcal{Q}^-_{s}$ has a $3_2$ cycle with $a,b$ descending from the hybrid node, and possibly  other cycles. We may suppress all $2_1$- and $2_3$-cycles by Corollary \ref{cor:2isgood} without affecting $CF(\mathcal{Q}^-_{s})$. By Lemmas  \ref{lem::formula2:2} and \ref{lem::formula43:1}, we may remove $2_2$- and $3_1$-cycles by expressing $CF(\mathcal{Q}^-_{s})$ as a convex sum of networks with a $3_2$-cycle, but one fewer cycle. Thus  $CF(\mathcal{Q}^-_{s})$ is a convex sum of points in $\mathcal{L}$, which lies in $\mathcal{L}$.  $\Box$\end{proof}

%

In the supplementary materials of \cite{Solis-Lemus2016} it is  stated that an unrooted quartet network $Q_{abcd}$ with a $3_2$-cycle can be always reduced to an unrooted quartet tree with some adjustment in the edge lengths. This is not true in general; that is, when $\{a,b,c,d\}$ satisfies the Cycle property it is not treelike. However, Proposition \ref{prop:3c2} indicates that sometimes unrooted quartet networks with $3_2$-cycles are treelike.\\

 To conclude this section, we show the Cycle property can give positive information about a network.\\

\begin{proposition}\label{prop::cpcycle}
	Let $\mathcal{Q}^-_s$ be an unrooted   level-1 quartet network on  a set of taxa  $s=\{a,b,c,d\}$. If $s$ satisfies the Cycle property, the unrooted quartet network $\mathcal{Q}^-_s$ contains either a $3_2$-cycle or a 4-cycle. 
\end{proposition}
\begin{proof}
 Proposition \ref{prop:Cfspreserv} shows that if $\mathcal{Q}^-_s$ has  neither a $3_2$-cycle nor a 4-cycle, the concordance factors of  $\mathcal{Q}^-_s$ are those of a tree. 
 $\Box$\end{proof}

\section{The Big Cycle property}\label{BC section}
In this section we investigate how to detect 4-cycles in a network from quartet concordance factors.\\  


Even though the Cycle property give us some information about an unrooted quartet network, it is not sufficient to tell us what the unrooted quartet network is.  This is shown by the following  Example, where a 4-cycle network  lead to identical concordance factors as those in Example \ref{ex::cpbc}.

\begin{example}\label{ex::cpbc4}
	Let  $\widetilde{Q}^-_{abcd}$ be the $4$-cycle unrooted quartet in Figure \ref{fig::4c1t}, where $\gamma=\frac{1}{2}$, $t_1=-\log\left(\frac{48}{49}\right)$ $=t_2$. By  the equations in \eqref{eq::4cyc}   the concordance factors are:
	\begin{align*}
	CF_{AB|CD}= \frac{32}{98},\text{ } CF_{AC|BD}=   \frac{33}{98},\text{ } CF_{AD|BC}=  \frac{33}{98},
	\end{align*}
	These agree with those of $\mathcal{Q}^-_{abcd}$ in Example \ref{ex::cpbc}.
\end{example} 

  This motivates the following definition.

\begin{definition}\label{bigcycle}
	Let $\mathcal{N}^+$ be a  metric rooted level-1 network on   $X$. We say that a  subset of four distinct taxa $\{a,b,c,d\}\subset X$ satisfies the  \textbf{Big Cycle}  property (denoted \textbf{ $BC$}) if  all the entries of $CF_{abcd}$ are different.
	 
	Let $\{a,b,c,d\}$ be a subset of taxa satisfying the $BC$ property. Denote by ${q^{BC}_{abcd}}$  the unrooted quartet corresponding to the smallest entry of $CF_{abcd}$. 
\end{definition}
For example, if $CF_{AB|CD}<CF_{AC|BD}<CF_{AD|BC}$, then
${q^{BC}_{abcd}}= AB|CD$.

Note that if $s$ satisfies the $BC$ property then $s$ satisfies the Cycle property but the Cycle property is weaker than the Big Cycle property.

\begin{proposition}\label{prop:Bc4cycle1dir}
	Let $\mathcal{Q}^-_s$ be an unrooted  level-1 quartet  network on  a set of taxa  $s=\{a,b,c,d\}$. If $s$ satisfies the $BC$ property, then the unrooted quartet network $\mathcal{Q}^-_s$ contains a 4-cycle. 
\end{proposition}
\begin{proof} 	
	By Proposition \ref{prop::cpcycle},  $\mathcal{Q}^-_s$ contains  either a $3_2$-cycle or  a 4-cycle, and by 
	  Proposition \ref{prop:3c2},  $\mathcal{Q}^-_s$ cannot have a $3_2$-cycle.    $\Box$\end{proof}

A converse of Proposition \ref{prop:Bc4cycle1dir} also holds, provided we include an assumption of generic parameters.
 
\begin{proposition}\label{prop:BCgenerically}
	Let $\mathcal{N}^+$ be a metric rooted level-1 on $X$ with $|X|\geq 4$.  Let $\{a,b,c,d\}\subset X$ such  that $\mathcal{Q}^-_{abcd}$ has a $4$-cycle. Then $\{a,b,c,d\}$ satisfies the Cycle property. Moreover, for generic numerical parameters on $\mathcal{N}^+$,  $\{a,b,c,d\}$ satisfies the $BC$ property. That is, for all numerical parameters except those in a set of measure zero, the $BC$ property holds.
\end{proposition}
\begin{proof}
	
	Let $s=\{a,b,c,d\}\subset X$ be such  that $\mathcal{Q}^-_{s}$ has a $4$-cycle. Without loss of generality suppose that  $c$ is the descendant of the hybrid node and the hybrid block $\{c\}$ of $\mathcal{Q}^-_{s}$ is adjacent to the $v$-blocks containing $a$ and $b$. Since $\mathcal{N}^-$ is level-1, the only other possible cycles in $\mathcal{Q}^-_{s}$ are $2_1$ or $2_3$-cycles. By Corollary \ref{cor:2isgood}, $CF(\mathcal{Q}^-_{s})=CF(Q')$, where $Q'$ is the network obtained after suppressing all cycles other than the $4$-cycle. Note that $Q'$ is the network shown in Figure \ref{fig::4c1t}, and by equations \eqref{eq::4cyc},  $CF(Q')$  depends only on the length of the non-hybrid edges in the $4$-cycle and the $\gamma$ parameter  of the hybrid edges of $\mathcal{Q}^-_{s}$. Moreover, equations \eqref{eq::4cyc} show that $\{a,b,c,d\}$ satisfies the Cycle property.
	
	When $\mathcal{Q}^-_{s}$ is obtained from $\mathcal{N}^-$, the lengths  of the edges of  $\mathcal{Q}^-_{s}$  are the sum of edge lengths from $\mathcal{N}^-$. Let $\Theta_{\mathcal{N}^-}=(0,\infty)^m\times [0,1]^h$ be the numerical parameter space for $\mathcal{N}^-$ and let $\Theta'_{s}=(0,\infty)^2\times [0,1].$ 
	Thus we can define a map $\nu_s:\Theta_{\mathcal{N}^-}\to \Theta'_s$ such that for any metric $(\lambda, \gamma)$ of $\mathcal{N}^-$, $\nu_s((\lambda, \gamma))$ encodes the edge length of the non-hybrid edges in the $4$-cycle and the $\gamma$ parameter of the hybrid edges. In particular this map is linear and surjective. 
	
	With $\chi_{s}=(0,1)^2\times [0,1]$, let $\eta:\Theta'_{s}\to \chi_{s}$ be defined as $\eta(l_1,l_2,\gamma)=(e^{-l_1},e^{-l_2},\gamma)$, so $\eta$ is  a biholomorphic function.   Defining $f:\chi_{s}\to \Delta_2$   by 
	$$f((L_1,L_2,\gamma))=(1-\gamma)(1-2L_1/3,L_1/3,L_1/3)+\gamma(L_2/3,L_2/3,1-2L_2/3),$$
	the quartet concordance factor map can be viewed as a composition
	$$\Theta_{\mathcal{N}^-} \xrightarrow{\nu_s}\Theta'_{s}\xrightarrow{\eta}\chi_{s}\xrightarrow{f}\Delta_2.$$

	It is straightforward to see that the image of $f$ restricted to $\gamma=0$ and $\gamma=1$ is the red (skewed) and blue (vertical) segments shown on the right of Figure \ref{fig:inversesimplex}. 

	Let  $V=V((x-z)(y-z)(x-y),x+y+z-1)$, that is, let $V$ be the algebraic variety composed of the points on which $(x-z)(y-z)(x-y)$ and $x+y+z-1$ are zero, as depicted  on the right of Figure \ref{fig:simplex1CP3}. Observe that $V$ is the points in $\Delta_2$ that, if interpreted as concordance factors,  would \emph{not} satisfy the $BC$ property. 

	Since $f$ is a polynomial map whose image is not contained in $V$, the pre-image of $V$ under $f$ is contained in a proper sub-variety of $\chi_s$, and therefore $f^{-1}(V)$ has measure zero in $\chi_{s}$. Since $\eta$ is biholomorphic, then $\eta^{-1}(f^{-1}(V))$ has measure zero. Since $\nu$ is linear surjective, then  $\nu^{-1}(\eta^{-1}(f^{-1}(V)))$ has measure zero. 
	Thus generic points in $\Theta_{\mathcal{N}^-}$ are mapped to concordance factors satisfying the $BC$ property.  $\Box$\end{proof}

 To better understand the geometry of the map $f$ in this proof, let   $s=\{a,b,c,d\}$ be a subset of  four distinct taxa satisfying the $BC$ property.  Figure  \ref{fig:inversesimplex} depicts the subset of $\chi_{s}$ that is mapped by $f$ to those segments of the shaded triangle inside $\Delta_2$. The interior of $\chi_{s}$ is mapped to the interior of the shaded triangle. 


\begin{figure}\begin{center}
	 
	\includegraphics[scale=.3]{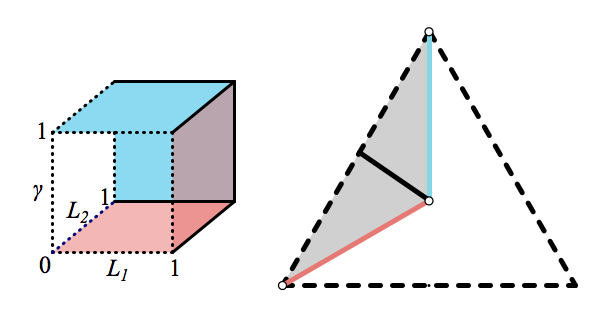}
	\caption{The function $f$ maps the cube $\chi_s$ (left) to $\Delta_2$ (right).  The blue facets (rear and top) of the cube are mapped by $f$ to the blue (vertical) segment  and the red facets (bottom and right) to the red  (skewed) segment.  The full cube is mapped onto the shaded triangle, giving all concordance factors with a 4-cycle as in Figure \ref{fig::4c1t}. The three line segments, two on the boundary of and one within the shaded triangle, are comprised of points not satisfying the $BC$ property.}\label{fig:inversesimplex}
 
\end{center}\end{figure}



The following Theorem follows immediately from Proposition  \ref{prop:BCgenerically} and Proposition \ref{prop:Bc4cycle1dir}.

\begin{theorem}\label{thm::twocor}
	Let $\mathcal{N}^+$ be a metric rooted level-1   network    on  $X$ with $|X|\geq 4$ and  $\{a,b,c,d\}\subset X$. For generic numerical parameters, $\{a,b,c,d\}$ satisfies the $BC$ property if and only if  $\mathcal{Q}^-_{abcd}$ has a $4$-cycle. 
\end{theorem}

Theorem  \ref{thm::twocor} and Proposition  \ref{prop::cpcycle}, yield the following.

\begin{corollary}\label{cor::Cpgen}
	Let $\mathcal{N}^-$ be a metric unrooted level-1 network on  $X$ and let $s=\{a,b,c,d\}$ be a set of distinct taxa in $X$. Then if $s$ satisfies the Cycle property but not the $BC$ property for generic parameters,  then $\mathcal{Q}^-_{s}$ contains a $3_2$-cycle. 
\end{corollary}
 
The converse of Corollary \ref{cor::Cpgen} does not hold, as pointed out by Proposition \ref{prop:3c2}.\\
 
If a set of 4 taxa satisfy the $BC$ property, we can deduce some finer  information about the 4-cycle on the unrooted quartet network and a larger network, as proved in the following.
 
\begin{proposition}\label{prop:isbc}Let $\mathcal{N}^-$ be a metric unrooted level-1 network on   $X$ and let $\{a,b,c,d\}\subseteq X$   satisfy the $BC$ property, so $\mathcal{Q}^-_{abcd}$ contains a $4$-cycle $C_v$. Then  ${q^{BC}_{abcd}}=AC|BD$ if and only the $v$-blocks of $\mathcal{Q}^-_{abcd}$ containing $a$ and $c$ are not adjacent.
\end{proposition}
\begin{proof}
	Let $Q=\mathcal{Q}^-_{abcd}$. Since $\mathcal{N}^-$ is level-1 the only possible cycles in $Q$,  other than $C_v$, are $2_1$ and $2_3$-cycles. Let $Q'$ be the network obtained after suppressing all $2_1$ and $2_3$-cycles, so $Q'$ has only a four cycle. By Corollary \ref{cor:2isgood},  $CF(Q)=CF(Q')$.  Thus by equations \eqref{eq::4cyc}, we obtain the desired result.  $\Box$\end{proof}

\begin{lemma}\label{cfs}Let $\mathcal{N}^-$ be a metric unrooted level-1 network on $X$ with generic numerical parameters. There exists  $\{a,b,c,d\}\subseteq X$ satisfying the $BC$ property if and only if $\mathcal{N}^-$ contains a cycle $C_v$ of size $k\geq 4$ with one of these taxa is in the hybrid block, and the others in distinct $v$-blocks on $\mathcal{N}^-$.
\end{lemma}
\begin{proof}
	Suppose that $\mathcal{N}^-$ has a cycle of  size $k$ for some $k\geq 4$ with hybrid node $v$. Choose four taxa $\{a,b,c,d\}$, such that $a$ is in the hybrid block and $a,b,c$ and $d$ are in distinct $v$-blocks.  This set of taxa induces a unrooted quartet network with a 4-cycle, and so by Theorem \ref{thm::twocor} this set of taxa satisfies the $BC$ property for generic parameters.   
	Suppose conversely, that there exists $\{a,b,c,d\}$ satisfying the $BC$ property. By Theorem \ref{thm::twocor}, $\mathcal{Q}^-_{abcd}$ has a 4-cycle, so  $\mathcal{N}^-$ has a cycle of at least size four and one of these taxa is a descendant of the hybrid node. Since the other taxa are in distinct $v$-blocks of $\mathcal{Q}^-_{abcd}$, they must be in distinct $v$-blocks of $\mathcal{N}^-$.  $\Box$\end{proof}

For a  level-1 metric unrooted network $\mathcal{N}^-$, let $S$ be the collection of sets of 4 distinct taxa satisfying the $BC$ property and $V_H$ be the set of hybrid nodes. We observe that for any $s\in S$, there is a natural map $\psi: S\mapsto V_H$, where $\psi(s)=v$ if $v$ is the hybrid node associated to the cycle of size 4 in $\mathcal{Q}^-_s$. In this case we say that  $s$  \textit{determines}   the hybrid node $v$. 

\begin{lemma}\label{3tax}
	Let $\mathcal{N}^-$ be a metric unrooted level-1 network and let $\{a,b,c,d\}$ and $\{a,b,c,e\}$  be  subsets of the taxa satisfying the $BC$ property.  The set  $\{a,b,c,d\}$ determines $v$ if and only if  $\{a,b,c,e\}$  determines $v$.
\end{lemma}
\begin{proof}
Let  $\{a,b,c,d\}$ determine $v$,  $\{a,b,c,e\}$ determine $u$, and suppose that $u\neq v$. Let $C_v$  and $C_u$ the cycles  in  $\mathcal{N}^-$ containing $v$ and $u$ respectively, so  $C_u$ and $C_v$ do not share edges. Since  $\{a,b,c,d\}$ satisfies the $BC$ property, by Lemma \ref{cfs}, $a$, $b$, $c$, and $d$ belong to different $v$-blocks, so that  in $\mathcal{N}^-\smallsetminus E(C_v)$ the taxa $a$, $b$ and $c$ are in different connected components.  Since $\mathcal{N}^-$ is level-1, $C_u$ is in one of the connected components of  $\mathcal{N}^-\smallsetminus E(C_v)$, say $\mathcal K$.  In particular note that all the taxa not in  $\mathcal K$ are in the same $u$-block. But at least two of $a,b$ and $c$ are not in  $\mathcal K$, so at least two of $a$, $b$  and $c$ are in the same $u$-block. This contradicts Lemma \ref{cfs}, so $u=v$.  $\Box$\end{proof}
 
Interestingly,  under the NMSC the ordering of quartet concordance factors is insufficient to identify the hybrid node of cycles of size $4$. For example, the networks shown in Figure \ref{fig:4cyc} all have the same ordering of their concordance factors despite different hybrid nodes. The concordance factors for all those networks have the same values:
\begin{figure}\begin{center}
	\includegraphics[scale=.28]{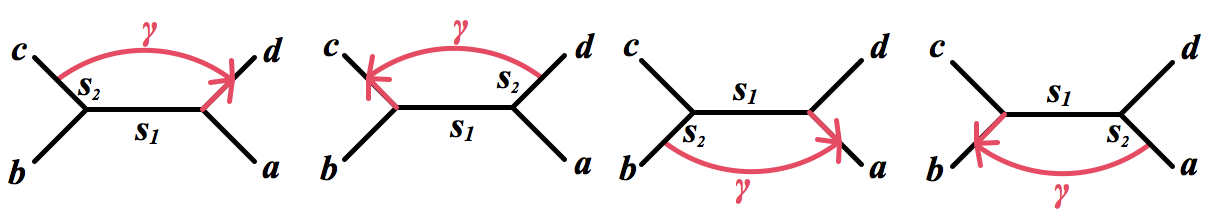}
	\caption{Four unrooted metric level-1  quartet networks with the same concordance factors.}\label{fig:4cyc}
\end{center}\end{figure}
 
\begin{align*}
CF_{ab|cd}&= (1-\gamma)\left(1-\frac{2}{3}e^{-s_1}\right)+\gamma\left(\frac{1}{3}e^{-s_2}\right)  ,\\
CF_{ac|bd}&= \left(1-\gamma\right)\left(\frac{1}{3}e^{-s_1}\right)+\gamma\left(\frac{1}{3}e^{-s_2}\right)  ,\\
CF_{ad|bc}&= \left(1-\gamma\right)\left(\frac{1}{3}e^{-s_1}\right)+\gamma\left(1-\frac{2}{3}e^{-s_2}\right)   .\\
\end{align*}
 
Figure \ref{fig:insidesimp} shows the 4-cycle network topologies drawn in the regions of $\Delta_2$ which their concordance factors fill.   In each case it does not matter which of the cycle nodes is the hybrid node; all those unrooted quartet networks give concordance factors that fill the that region.\\

\begin{figure}\begin{center}
	\includegraphics[scale=.3]{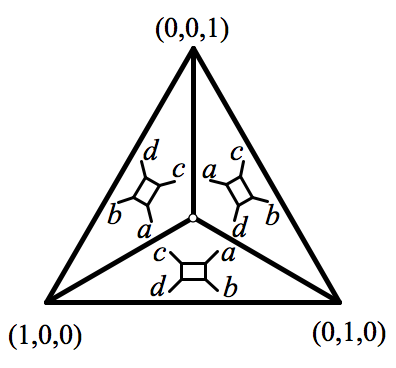}
	\caption{Each section of the simplex is depicted with an unrooted quartet network topology whose image under the concordance factor map fills that  region, independent of the placement of the hybrid node.}\label{fig:insidesimp}
\end{center}\end{figure}

\section{Identifying cycles in networks}\label{sec final}

Having  shown that the $BC$ property can detect the existence of 4-cycles in networks, for generic parameters, we are poised to prove  our main result. Our arguments now are mainly combinatorial.\\

Given a network $\mathcal{N}^+$ on $X$, let $S$ denote the set of 4-taxon subsets of $X$  satisfying the $BC$ property. 


\begin{lemma}\label{netblo}Let $\mathcal{N}^+$ be a metric rooted level-1 network on $X$. Then under the NMSC model with generic parameters the $4$-network blocks of  $\mathcal{N}^+$ can be determined from the set $S$.
\end{lemma}
\begin{proof}
	If $|X|<3$ there is nothing to prove. The case $|X|=4$ follows from Proposition \ref{prop:Bc4cycle1dir}, so we assume $|X|\geq 5$. 
	By Lemma \ref{cfs}, for any $\{a,b,c,d\}\in S$  each taxon $a$, $b$, $c$, $d$ must belong to a different $4$-network block. Let 
	$$Y_a=\bigcup_{   \scaleto{\{s\in S \mid a\in s\}}{9pt}} s\smallsetminus\{a\}$$
	Then $Y_a$ is the complement of the $4$-network block containing $a$. To see this,  note that for any taxon $b$  that does not belong to the  $4$-network block of $a$, by Lemma \ref{2taxnetblo}, there exists a cycle $C_v$ of size at least   $4$ such that $a$ and $b$ are in different $v$-blocks. Now choose any two different taxa $c$ and $d$,  such that all taxa $a$, $b$, $c$, $d$ are in different $v$-blocks  and one of $a$, $b$, $c$ or $d$ is in the $v$-hybrid block. Then $\{a,b,c,d\}\in S$, and thus $b\in Y_a$. 
	
	It follows that $X\smallsetminus Y_x$ is the $4$-network block containing taxon $x$. Since $x$ was arbitrary, all 4-network blocks can be determined.  $\Box$\end{proof}

\begin{lemma}\label{blov}Let $\mathcal{N}^+$ be a metric rooted level-1 network on $X$ with cycle $C_v$ of size $k_v\geq 4.$ Then for generic parameter choices, the  $v$-blocks  and the size $k_v$ can be identified from the set $S$. If $k_v\ge 5$  the $v$-hybrid block can also be identified. 
\end{lemma} 
\begin{proof}
	
	Let $\{a,b,c,d\}\in S$ and let $v$ be the hybrid node  determined by it. By Lemma \ref{cfs}, each of these taxa belongs to a different $v$-block, and hence to a different $4$-network block. 
	Denote by $A,B,C,D$ the $v$-blocks containing $a,b,c$ and $d$ respectively. 
	
	Let $Z_{abc}$ be the set of all taxa $e$ such that  $\{a,b,c,e\}\in S$. By Lemma \ref{3tax}, all such $\{a,b,c,e\}\in S$ determine the same hybrid node $v$. Consider now  $Z_{bcd},$ $Z_{acd}$ and $Z_{abd}$.   If $k_v=4$, then, by the last statement of Lemma \ref{cfs}, $Z_{abc}=D$, $Z_{bcd}=A$, $Z_{acd}=B$ and $Z_{abd}=C$, so all pairwise intersections of $Z_{abc},$ $Z_{bcd},$ $Z_{acd},$ $Z_{abd}$ are empty.   If $k_v>4$, then, again by Lemma \ref{cfs}, for some distinct taxa $i,j,k\in\{a,b,c,d\}$,  $Z_{ijk}$ is the $v$-hybrid block, and for any $l,m,n\in\{a,b,c,d\}$ with $\{l,m,n\}\neq \{i,j,k\}$,  $Z_{lmn}=(L\cup M\cup N)^c$.  Note that $Z_{ijk}\cap Z_{lmn}=\emptyset$ since one of $L,M,N$ is the $v$-hybrid block. Since $Z_{lmn}$ contains at least one $v$-block other than $A$, $B$, $C$ or $D$, for any $l',m',n'\in\{a,b,c,d\}$, with $\{l',m',n'\}\neq\{i,j,k\}$, $Z_{lmn}\cap Z_{l'm'n'}\neq \emptyset$. Hence we can determine whether $k_v> 4$ or $k_v=4$: if all pairwise intersection of $Z_{abc},$ $Z_{bcd},$ $Z_{acd},$ $Z_{abd}$ are empty then $k_v=4$, else $k_v>4$. If $k_v>4$ we can determine the hybrid block, by noting which of the sets $Z_{abc},$ $Z_{bcd},$ $Z_{acd},$ $Z_{abd}$ has empty intersection with any other set in this family. At this point we have determined either that $k_v=4$ and all $v$-blocks, or that $k_v>4$ and the hybrid block.

	In the case $k_v> 4$, without loss of generality, suppose that $A$ is the $v$-hybrid block. Let $y\notin Z_{abc}=(A\cup B\cup C)^c$, so $y$ is in one of $A$, $B$ and $C$. For some $u,w\in\{a,b,c\}$,  $s'=\{y,u,w,d\}\in S$, which shows $y$ and the taxon $g\in\{a,b,c\}\smallsetminus \{u,w\}$  are in the same $v$-block. Thus we can determine $A$, $B$ and $C$.
 
	Note that for any taxon $x$ that is not in any of $A,$ $B$ or $C$, then $s=\{a,x,b,c\}\in S$. Since $s$ determines $v$, following the steps of the last paragraph identifies the $v$-block that contains $x$. Therefore  all $v$-blocks can be determined, and thus $k_v$ as well.  $\Box$\end{proof}
 
\begin{lemma}\label{orderofcycles}
	Let $\mathcal{N}^+$ be a metric rooted level-1 network on $X$. Then for any hybrid node $v$ with $k_v\geq 4$ the order of the $v$-blocks in the cycle can be determined from the ordering of the concordance factors.
\end{lemma}
\begin{proof}
	If $k_v=4$, the claim is established by Proposition \ref{prop:isbc}. Now suppose that $k_v>4$, so by Lemma \ref{blov} we know the $v$-hybrid block.  Let $A_1,...,A_{k_v}$ be the $v$-block partition with $A_1$ the $v$-hybrid block. Let $a_i\in A_i$ be an element of the $i$-th $v$-block. By Proposition \ref{prop:isbc},  $A_1$ and $A_j$ are adjacent if and only if $q^{BC}_{a_1a_jxy}\neq a_1a_j|xy $ for any distinct $x,y\in\{a_2,...,a_{k_v}\}\smallsetminus\{a_j\}$. Thus we can identify the two $v$-blocks adjacent to $A_1$. Suppose that such $v$-blocks are $A_p$ and $A_q$. We find the other $v$-block adjacent to $A_q$  from $\{ q^{BC}_{a_1a_pa_ja_m} \}$ for all distinct $j,m\in\{2,3,4,...,k_v\}\smallsetminus\{p,q\}$.  This is, $A_q$ and $A_j$ are adjacent if and only if $q^{BC}_{a_1a_ja_px}\neq a_1a_j|xa_p $ for any distinct $x\in\{a_2,...,a_{k_v}\}\smallsetminus\{a_p,a_q,a_j\}$ and $j\neq 1,p,q$. 	Continuing in this way, the full order of blocks around the cycle can be determined.  $\Box$\end{proof}
We reach the main result.

\begin{theorem}\label{thm::main}
	Let $\mathcal{N}^+$ be a metric rooted level-1 network on $X$.  Then under the NMSC model, for generic parameters, the collection of orderings of  quartet concordance factors identifies the unrooted semidirected topological network $\widetilde{\mathcal{N}}$ obtained from $\mathcal{N}^-$ by suppressing all 2- and 3-cycles, and directions of hybrid edges in 4-cycles, while retaining directions of hybrid edges of $k$-cycles for $k\ge 5$.
\end{theorem}
\begin{proof}
	We proceed by induction in the number of cycles of size $\geq4$. Suppose there are no such cycles.Then every induced quartet tree will have no cycle of size $4$, and the ordering of the concordance factors determines the topology of the quartet tree obtained by suppressing all $2$- and $3$-cycles. These then determine the topology $\widetilde{\mathcal{N}}$ by a standard result \cite{Semple2005}.
	
	Suppose there is exactly one cycle of size at least 4. Then there is just one hybrid node $v$ in $\mathcal{N}^-$ with $k_v\ge 4$. By Lemmas \ref{blov} and \ref{orderofcycles} we can determine the size $k_v$ of the cycle, the $v$-blocks and the order of the $v$-blocks in the cycle. If $k_v\geq 5$ we can identify the hybrid node $v$ and thus  identify the direction of the hybrid edges. 
	Let $P_u$ be a $v$-block where $u$ is a node in $C_v$, and $q\in X\smallsetminus P_u$.  Let $\mathcal{K}$ be the induced network on $P_u\cup \{q\}$ with all 2-cycles and 3-cycles suppressed.  
	Note that $\mathcal{K}$  is a tree, and the quartet concordance factors for taxa in $P_u\cup{q}$ identify its  topology. Viewing $q$ as an outgroup of $P_u$, induces a rooted tree on $P_u$. The root can then be joined with an edge to $u$. Doing this for all $v$-blocks establishes the claim.
		
	Now suppose that the result is true for networks with $l$ cycles of size at least $4$, and $\mathcal{N}^-$  contains $l+1$ such cycles. We can first determine all $4$-network blocks and the $v$-blocks and its cycle order for every cycle of size at least 4 by Lemmas \ref{netblo},  \ref{blov}, and \ref{orderofcycles}. 	Following Definition \ref{treeofcycles}, consider $\mathcal{T}$, the tree of cycles of $\widetilde{\mathcal{N}}$.  A leaf of $\mathcal{T}$  arises from a cycle $C_v$ on $\mathcal{N}^-$ if and only if all $v$-blocks but one are 4-network blocks. We may therefore determine the $v$-blocks of some cycle $C_v$ that is a leaf of $\mathcal{T}$. 
	
	Let $u$ be the vertex in $C_v$ associated to the $v$-block that is not a $4$-network block. Note that $\widetilde{\mathcal{N}}\smallsetminus\{u\}$ is a disconnected graph, with two connected components $\widetilde{\mathcal{N}_1}$ and $\widetilde{\mathcal{N}_2}$. Let $\widetilde{\mathcal{N}_1}$ be the component containing all nodes of $C$ except $u$, and  $S_i$  the set of taxa on $\widetilde{\mathcal{N}_i}$, $i\in\{1,2\}$. Let $s_i\in S_i$. Then 
	$\mathcal{N}^-_{S_i\cup\{s_j\}}$ for $i,j\in\{1,2\}$,  $i\neq j$, has at most $l$ cycles of size at least $4$. By the induction hypothesis  we can determine the semidirected topological network $\mathcal{N}_i$ obtained from  $\mathcal{N}^-_{S_i\cup\{s_j\}}$ by suppressing all 2- and 3-cycles, and directions of the hybrid edges in 4-cycles, while retaining directions of the hybrid edges of $k$-cycles for $k\ge 5$.  We obtain $\widetilde{\mathcal{N}}$ by identifying $s_1$ in $\mathcal{N}_2$ with $s_2$ in $\mathcal{N}_1$ and suppressing that node. $\Box$\end{proof}

Figure \ref{fig::finalex} shows a phylogenetic metric rooted network $\mathcal{N}^+$  and  $\widetilde{\mathcal{N}}$, the unrooted semidirected topological network which is identified by Theorem \ref{thm::main}. The cycle colored in green is a $4$-cycle and, though, its hybrid node is not identified from quartet concordance factors. However, its hybrid node has to be such that  $\widetilde{\mathcal{N}}$  is induced from a rooted network. Thus the node labeled $x$ in Figure \ref{fig::finalex} cannot be the hybrid node. This illustrates that although we cannot always identify the hybrid node on  $4$-cycles, sometimes the structure of the resulting network $\widetilde{\mathcal{N}}$ restricts the possible nodes for  its placement.   
\begin{figure}\begin{center}
	\includegraphics[scale=.22]{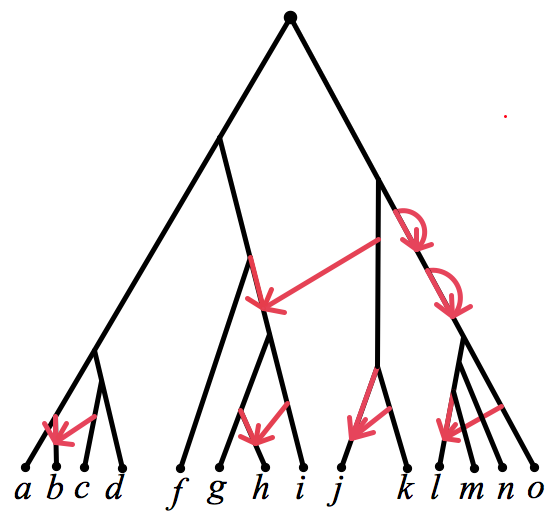}\includegraphics[scale=.22]{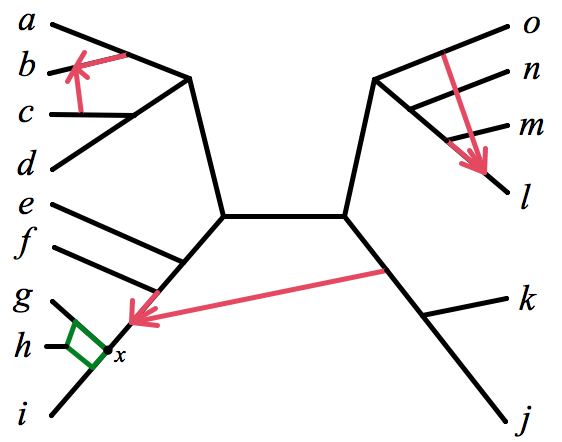}
	\caption{A rooted metric phylogenetic network  $\mathcal{N}^+$ (left) and the network structure  $\widetilde{\mathcal{N}}$ (right) that can be identified by Theorem \ref{thm::main}. The 4-cycle on the network in the right has 3 different candidates for the  hybrid node.}\label{fig::finalex}
\end{center}\end{figure}

\section{Further results on $3_2$-cycles}\label{sec::cycleprop}

Under some special circumstances, for example when a set of taxa satisfy the Cycle property but not the $BC$ property, it is possible to detect further information about the topology of the network than that given in Theorem \ref{thm::main}. For instance, some 3-cycles are identifiable under such hypothesis. In this section, we discuss these extensions briefly, as  it is difficult to formulate general statements on identifiability.\\

Recall that a $3_2$-cycle may lead to concordance factors satisfying the Cycle property, but it need not, as shown in Proposition \ref{prop:3c2}. There is a full-dimensional subset of parameters space on which concordance factors indicate a $3_2$-cycle and another in which it fails to. Nonetheless, the following gives a positive, but limited, identifiability result. 

\begin{proposition}
	Let $\mathcal{N}^+$ be a metric rooted level-1 network on $X$  and suppose $\{a,b,c,d\}\subset X$  satisfies  the Cycle property but not the $BC$ property.  Then under the NMSC model, for generic parameters, if there is no taxon $e\in X$ such that $\{i,j,k,e\}$ satisfies the $BC$ property for any distinct $i,j,k\in \{a,b,c,d\}$ then  $\mathcal{N}^-$ contains a $3$-cycle with at least two descendants of the hybrid node.
\end{proposition}
\begin{proof}
	Since $\{a,b,c,d\}\subset X$  satisfy  the Cycle property but not the $BC$ property, by Proposition \ref{prop::cpcycle},  there is a $3_2$-cycle in $\mathcal{Q}^-_{abcd}$. Thus three taxa of  $a,b,c,d$ are in distinct $v$-blocks in $\mathcal{Q}^-_{abcd}$. This implies  that   there exists a cycle $C_v$ in $\mathcal{N}^-$ where three taxa of  $a,b,c,d$ are in distinct $v$-blocks.  	Since $\{i,j,k,e\}$ does not satisfy the $BC$ property for any  distinct $i,j,k\in \{a,b,c,d\}$, this implies $C_v$ is not a $k$-cycle for $k\ge 4$. Thus by Proposition \ref{prop:3c2}, $C_v$ has size $3$ and at least two of $a$, $b$, $c$, $d$ descend from $v$. 
 $\Box$\end{proof}

Let $\mathcal{Q}^-_{abcd}$ be an unrooted  level-1 quartet network  where $\{a,b,c,d\}$ satisfies the Cycle property but not the $BC$ property.   It can be shown that if, for example, the smallest entry in $CF_{abcd}$ is the one corresponding to the quartet $AB|CD$, then either $a,b$ or $c,d$ are in the $v$-hybrid block. This proof is very similar to that of Proposition \ref{prop:isbc}. \\

Let $\mathcal{N}^+$ be a network such that  $\widetilde{\mathcal{N}}$ (the network obtained from $\mathcal{N}^+$ in Theorem \ref{thm::main}) is as shown in Figure \ref{fig:magicquintet}. Observe that $\{a,b,c,d\}$  satisfies the $BC$ property by Theorem \ref{thm::twocor}. If $\{a,e,b,d\}$ satisfies the Cycle property, then the  following Proposition indicates the hybrid node can be determined. 
\begin{figure}\begin{center}
	\includegraphics[scale=.25]{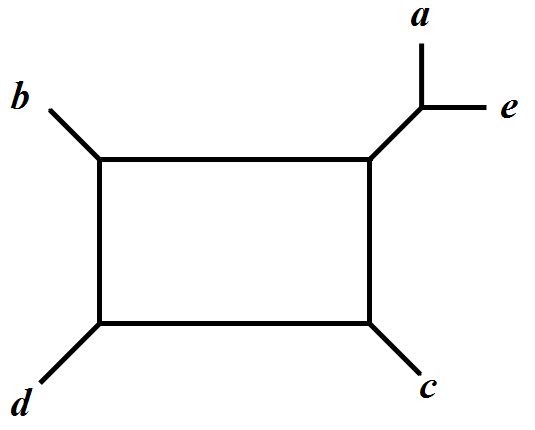}
	\caption{A network  $\widetilde{\mathcal{N}}$ with a four cycle such that if $\{a,b,c,e\}$ satisfies the Cycle property, the hybrid block can be detected.}\label{fig:magicquintet}
\end{center}\end{figure}

\begin{proposition}\label{lem::cpto4c}
	Let $\mathcal{N}^+$ be a metric rooted level-1  network on  $X$ and let $C_v$ be a $4$-cycle in $\mathcal{N}^-$. Let $a,b,c,d\in X$ be in different $v$-blocks in $\mathcal{N}^-$.  Suppose under the NMSC model, for generic parameters, for distinct $i,j,k\in \{a,b,c,d\}$, there exists a taxon $e \in X$ such that $\{i,j,k,e\}$ satisfies the Cycle property but not the $BC$ property. Then the $v$-block containing $e$ is the $v$-hybrid block.  
\end{proposition}
\begin{proof}
	Without loss of generality suppose that $i=a$, $j=b$ and $k=c$. Note  that $e$ is not in the same $v$-block as $d$, otherwise $\{a,b,c,e\}$ would satisfy the $BC$ property. Thus $e$ is the same $v$-block as $a,b$ or $c$. Without loss of generality suppose that is in the same $v$-block as $a$. Thus $\{e,b,c,d\}$ satisfies the $BC$ property and by Theorem \ref{thm::main} the order of the cycle can be determined. Without loss of generality suppose that the order is the one as in Figure \ref{fig:magicquintet}. By	Lemma  \ref{3tax},  $\{a,b,c,d\}$ and  $\{e,b,c,d\}$ determine the same hybrid node $v$.  Since $\{a,b,c,e\}$ satisfies the Cycle property, Corollary \ref{cor::Cpgen}  shows $\mathcal{Q}^-_{abce}$ has a $3_2$-cycle.  The $4$-cycle in $\mathcal{Q}^-_{abcd}$ and the $3$-cycle in $\mathcal{Q}^-_{abce}$ have to have the same hybrid edges, otherwise the level-1 condition would be violated.  Observe that the only possibility for $\mathcal{Q}^-_{abce}$  having a $3_2$-cycle  is if $e$ and $a$ are in the hybrid block. 	  $\Box$\end{proof}
 

\section{Appendix}\label{app}
Here, Proposition \ref{prop::ancestralinduced} of Section \ref{sec def} is proved. The argument uses the following.
\begin{lemma}\label{lem::trekmrca}
	Let $\mathcal{N}^+$ be a (metric or topological) rooted network on  $X$  and let $Z\subset X$. For any edge $e$ below MRCA$(Z)$, with a descendant in $Z$, there are $x,y\in Z$ such that $e$ is in a simple trek in $\mathcal{N}^+$ from  $x$ to $y$ whose edges are below MRCA$(Z)$. 
\end{lemma}
\begin{proof}
	Let  $x\in Z$ be below $e$. By Lemma \ref{lem::mrcaxy}  there exists $y\in Z$ with MRCA($x,y$)  above $e$.
	
	 Suppose $y$ is not below $e$. Let $P_x$ be a path from MRCA($x,y$) to $x$ containing $e$ and let $P_y$ be a path from MRCA($x,y$) to $y$. Let $u$ be the minimal node in the intersection of $P_x$ and $P_y$. Since $y$ is not below $e$, $u$ cannot be below $e$. Then the subpath of $P_x$ from $u$ to $x$, which  contains $e$, and the subpath of $P_y$ from  $f$ to $y$ form a simple trek containing $e$.

 		Now assume $y$ is below $e$.  Since $e$ is below MRCA($x,y$), there exists a path   from  MRCA($x,y$) to one of $y$ or $x$ that does not pass through the child of $e$. Without loss of generality suppose  such a path $P_y$ goes from MRCA($x,y$) to $y$.  	Let $P_x$ be a path from  MRCA($x,y$) to $x$ that passes through $e$. Let $A=A(P_x,P_y)$ be the set of nodes above $e$,  common to  $P_y$ and $P_x$.   Let $a\in A$  be the minimal node  in  $A$.

 			Let $B(P_y,P_x)$ be the set of nodes below $e$,  common to  $P_y$ and $P_x$.  We may assume that we choose $P_x$ and $P_y$ such that $B=B(P_y,P_x)$ has minimal cardinality. If $B=\emptyset$ then the desired trek is easily constructed, with top $a$. So suppose $B\neq \emptyset$ has minimal element $b^-$ and maximal element $b^+$.  We are going to contradict the minimality of $B$. Note that $b^+$ must be  the hybrid node of a cycle containing $e$ (see Figure \ref{fig::cases} for a graphical reference). 
	
	Since $b^-$ is not the MRCA($x,y$), there exists a path $P^*$ from MRCA($x,y$) to one of $x$ or $y$ that does not pass through $b^-$.  Note that $P^*$ has to intersect at least one of $P_y$ or $P_x$ at an internal node below $b^-$. Let $C_1$  be the set of nodes below $b^-$, common to $P^*$ and $P_y$ and let $C_2$ be the set of nodes below $b^-$, common to $P^*$ and $P_y$. Let $c$ be the maximal node in $C_1\cup C_2$. We can assume, without loss of generality, that $c$ is in $P_y$. This is because if instead, $c$ were in $P_x$ , we can construct paths $P_x'$ and $P_y'$  where $P_i'$ contains all the edges in $P_i$ above $b^-$ and all edges of $P_j$ below $b^-$ for $i,j\in\{x,y\}$, $i\neq j$. Note that $P_x'$  passes through $e$ and does not contains $c$, while $P_y'$ does not pass through $e$, contains $c$, and $B=B(P_y',P_x')$.

	  Denote by $W$ the set of nodes in $(P^*\cap P_y)\cup(P^*\cap P_x)$ and let $w$ be the minimal node of $W$  above $b^-$. Since $\mathcal{N}^+$ is binary,  $w$ cannot be  $a$ or $b^+$  (see Figure \ref{fig::cases} for a graphical reference). There are 5 different cases of the location of $w$ in the network composed by the paths $P_y$ and $P_x$. These are   
	\begin{enumerate}
		\item $w$ is in $P_y$, above $b^+$ but below $a$.
		\item $w$ is in $P_x$, above $b^+$ but below $e$.
		\item $w$ is in $P_x$, above $e$ but below $a$.
		\item $w$ is in one or more of $P_x$ or $P_y$, above $a$.
		\item $w$ is in one or more of $P_x$ or $P_y$, above $b^-$ but below $b^+$. 
	\end{enumerate}
	Figure \ref{fig::cases} depicts in gray the graph composed by the paths $P_y$ and $P_x$, and in black we see the possible subpaths of $P^*$ from $w$ to $c$. In any of case 1, 2 or 3 we can find a simple trek containing $e$ as depicted in Figure \ref{fig::case123} by choosing the appropriate edges, and thus $B$ was not minimal. For case 4 and 5 there are two possibilities; (i) $w$ is in both $P_y$ and $P_x$; (ii) $w$ is only in one of $P_y$ or $P_x$. For case 4 (i), the situation is simple, and we can find a simple trek  as depicted on the left in Figure \ref{fig::case45}. For case 4 (ii), we first find  the node in $A$ that is right above $w$. Then as depicted on the left of Figure \ref{fig::case45} we can find a simple trek. 
	
	For case 5 we do not find a simple trek directly, instead we construct two paths $P_1$ and $P_2$ from MRCA($x,y$) to  $x$, $y$ respectively, only one of which contains $e$ with at least one less node in $B(P_1,P_2)$ than $B$. For case 5 (i), we just take $P_1$ to be the same as $P_x$ and for $P_2$ we consider the same edges that are in $P_y$ above $w$, the edges below $c$, and the edges in $P^*$ between $w$ and $c$. For case 5 (ii), we assume without loss of generality that $w$ is in $P_x$. Let $b$ be the node in $B$ right above $w$. Let  $P_1$ be the path containing the edges in $P_x$ that are above $b$, the edges in $P_y$ that are below $b$ but above the node $b'\in B$ right below $w$, and at last the edges in $P_x$ below $b'$. Let $P_2$ the path containing the edges in $P_y$ that are above $b$, the edges in $P_x$ that are above $a$ but below $b$, the edges in $P^*$ that are above $c$ but below $w$ and at last the edges in $P_y$ that are below $c$. Figure \ref{fig::case45} (right) depicts $P_1$ (red) and $P_2$ (blue) for (i) and (ii). Since  $B(P_1,P_2)$ has at least one less node that $B$ and we assumed $B$, the minimality of $B$ is contradicted. $\Box$\end{proof}
	\begin{figure}\begin{center}
		\includegraphics[scale=.20]{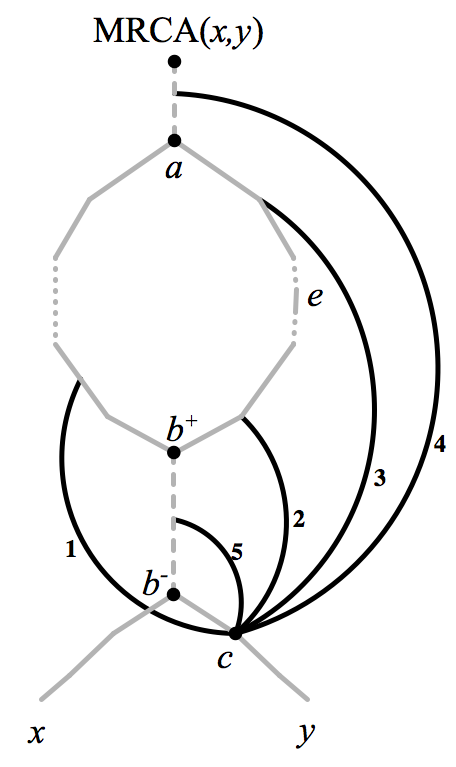}\caption{In gray we see the subgraph composed by $P$ and $P'$, the dashed edges represent that $P$ and $P'$ could intersect, the dotted segments represent  just a succession of edges. In black we see the different cases of the possible edges in $P^*$ above $b$ but below $a$.}\label{fig::cases}
	\end{center}\end{figure}
	\begin{figure}\begin{center}
		\includegraphics[scale=.20]{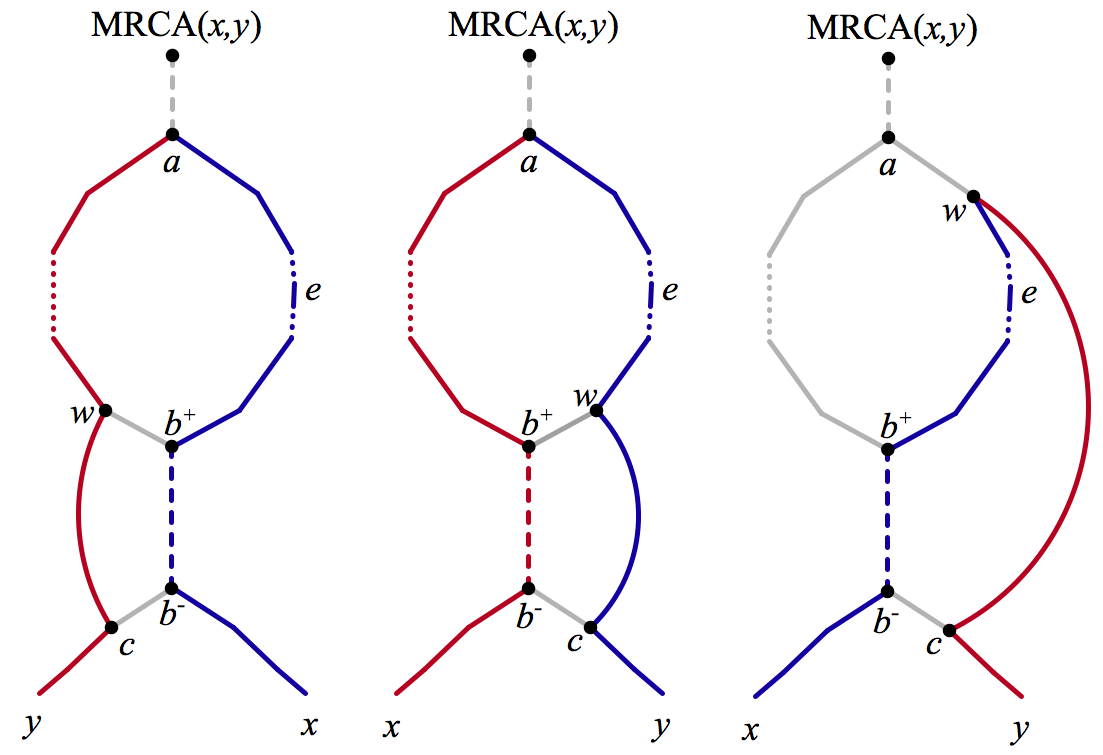}\caption{The treks in case 1 (left), case 2 (center), and case 3 (right).}\label{fig::case123}
	\end{center}\end{figure}
	\begin{figure}\begin{center}
		\includegraphics[scale=.192]{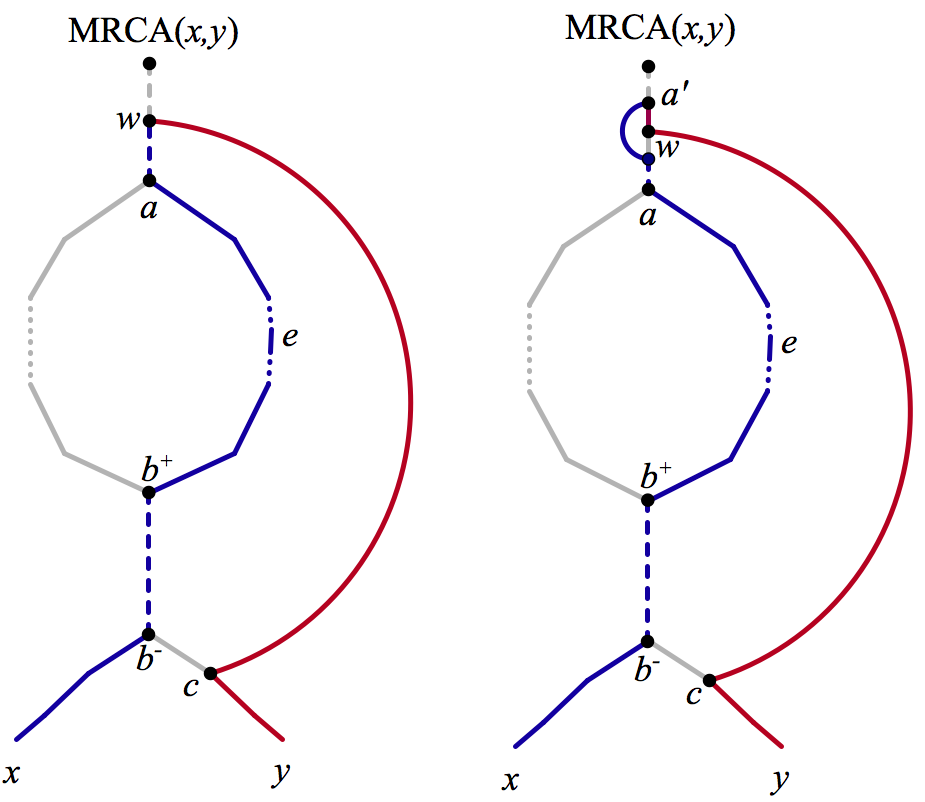}\includegraphics[scale=.205]{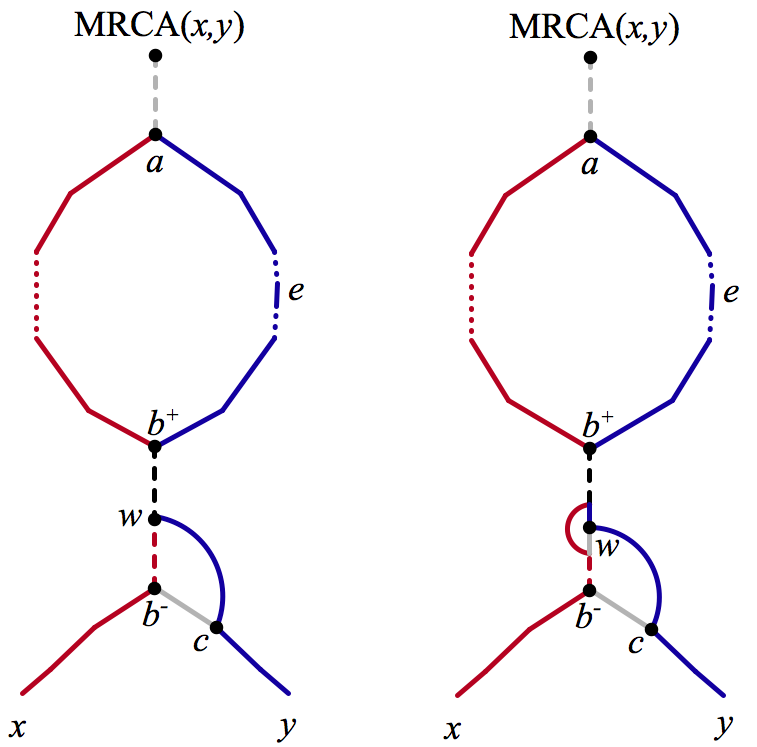}\caption{(Left) The treks in the two possibilities of case 4. (Right) The two possibilities of case 5, where the black segments represent possible edges red and blue at the same time. }\label{fig::case45}
	\end{center}\end{figure}
 
\begin{proof}[of Proposition \ref{prop::ancestralinduced}]
	Let $M^+=\mathcal{N}^\oplus_Z$.  Let $M^-$ be the graph obtained from $M^+$ by ignoring  the direction of all tree edges and then suppressing the  MRCA($Z,\mathcal{N}^+$), that is, the induced unrooted network from  $M^+$.  	Denote by $M'$ the graph obtained by ignoring all directions of the tree edges in $M^+$, so that by suppressing   degree two nodes of either $M^-$ or $M'$ gives  $(\mathcal{N}^+_Z)^-$. Let $K$ be the graph obtained by considering all the edges in simple treks in $\mathcal{N}^-$ from $x$ to $y$ for all $x,y\in Z$, so that   suppressing degree two nodes in $K$ gives  $(\mathcal{N}^-)_Z$. Showing  either $M'=K$ or  $M^-=K$, will prove the claim.\\ 
	
	First we show that if  MRCA($Z,\mathcal{N}^+$)$\neq$MRCA($X,\mathcal{N}^+$) then $M'=K$,   by arguing that   $M'$ and $K$ have the same edges. Let $e$ be an edge of $M'$. Since MRCA($Z,\mathcal{N}^+$)$\neq$MRCA($X,\mathcal{N}^+$), $M'$ is a subgraph of $\mathcal{N}^-$ and  $e$ is  directed in $M^+$. By Lemma \ref{lem::trekmrca}, $e$ is in a simple trek in $M^+$  from  $x$ to $y$, for some $x,y\in Z$. This trek induces a simple trek in $M'$  from   $x$ to $y$,  and therefore a simple trek in  $\mathcal{N}^-$ from   $x$ to $y$.    Thus  $e$ is in  $K$. 
	
	Now let $e$ be an edge of $K$. Then there exists a simple trek   $(\ol{P_1},\ol{P_2})$   in  $\mathcal{N}^-$  from  $x$ to  $y$, for some $x,y\in Z$ containing $e$.  Let $v=$top$(\ol{P_1},\ol{P_2})$ and let $T$  be the   sequence of  incident edges in  $\mathcal{N}^+$ from $x$ to $v$ conformed of edges inducing those in $\ol{P_1}$ and $\ol{P_2}$. Since $(\ol{P_1},\ol{P_2})$ is  simple, $T$ does not have repeated edges. Following $T$ in $\mathcal{N}^+$ from $x$ to $y$, edges are first transversed ``uphill" (in reverse direction) until there is a first ``downhill" edge $(u,w)$. The next edge in $T$  cannot be uphill, as otherwise it would be hybrid  and  $(\ol{P_1},\ol{P_2})$ would have not been a trek in $\mathcal{N}^-$. This argument applies for all  consecutive edges in $T$ until  we end at $y$. Thus there is a simple trek  $(\ol{P_1},\ol{P_2})$ from $x$ to $y$ in $\mathcal{N}^+$ with top $u$.  Note that $u$ must be below or equal to MRCA($Z,\mathcal{N}^+$) since otherwise the trek would not be simple. Moreover, $P_1$ and $P_2$ contain only edges in $M^+$ and thus in $M'$ after the directions of the tree edges is omitted. Thus  $e$ is  in $M'$, so $K=M'.$

If MRCA($Z,\mathcal{N}^+$)$=$MRCA($X,\mathcal{N}^+$) then $M^-=K$ follows from a straight forward modification of the previous argument to account for the suppression of MRCA$(z,\mathcal{N}^+)$ in both $M^-$ and $K$.  $\Box$\end{proof}
\begin{acknowledgements}
 The author deeply thanks  John A. Rhodes and  Elizabeth S. Allman for their technical assistance and suggestions during the development of this work. 
\end{acknowledgements}

\bibliographystyle{plain}
\bibliography{Hybridization}
\end{document}